\documentclass[11pt]{article}
\usepackage[vmargin=.65in,hmargin=.74in]{geometry}
\usepackage{amssymb,latexsym}
\usepackage{graphicx,amsmath,amssymb}
\usepackage{amsthm}
\usepackage{microtype}
\usepackage{enumitem}
\usepackage{setspace}
\usepackage{bussproofs}
\usepackage{mathrsfs}
\usepackage{authblk}
\usepackage{color}
\usepackage{tikz}
\usetikzlibrary{positioning,shapes.geometric,arrows.meta,decorations.pathreplacing}
\usepackage[section]{placeins}
\usepackage[font=scriptsize]{caption}
\usepackage{subcaption}
\usepackage{iddo}
\usepackage{fullpage}
\usepackage{kbordermatrix}
\usepackage{authblk}
\usepackage{lineno}

        \usepackage{nth}
\usepackage{intcalc}
\usepackage{etoolbox}
\usepackage{xstring}

\newcommand{\ECCCurl}[2]{http://eccc.hpi-web.de/report/\ifnumcomp{#1}{>}{93}{19}{20}#1/#2/}

\newcommand{\shortECCC}[2]{\texttt{\href{http://eccc.hpi-web.de/report/\ifnumcomp{#1}{>}{93}{19}{20}#1/#2/}{eccc:TR#1-#2}}}

\newcommand{\parseECCC}[1]{
\StrSubstitute{#1}{TR}{}[\tmpstring]%
\IfSubStr{\tmpstring}{/}{ 
\StrBefore{\tmpstring}{/}[\ecccyear]%
\StrBehind{\tmpstring}{/}[\ecccreport]%
}{
\StrBefore{\tmpstring}{-}[\ecccyear]%
\StrBehind{\tmpstring}{-}[\ecccreport]%
}%
\shortECCC{\ecccyear}{\ecccreport}}

\newtheorem{conjecture}[theorem]{Conjecture}
\newtheorem{problem}[theorem]{Problem}
\newtheorem{questn}[theorem]{Question}

\newtheorem*{definition*}{Definition}
\newtheorem*{lemma*}{Lemma}
\newtheorem*{claim*}{Claim}

\newcommand{\LinSys}{\mathsf{LinSys}}
\newcommand{\LinDags}{\mathsf{LinDags}}
\newcommand{\LinTrees}{\mathsf{LinTrees}}
\newcommand{\BinDags}{\mathsf{BinDags}}
\newcommand{\BinRegDags}{\mathsf{BinRegDags}}
\newcommand{\reslin}[1]{\mathsf{Res(lin_{#1})}}
\newcommand{\reslinneq}[1]{\mathsf{Res(lin^{\neq}_{#1})}}
\newcommand{\ECC}{\mathsf{ECC}}
\newcommand{\mc}[1]{\mathcal{#1}}
\newcommand{\bmodels}{\models_{0\text{-}1}}
\newcommand{\spanv}[1]{\langle #1 \rangle}
\newcommand{\CE}{(q+1)\ln q}
\newcommand{\CI}{6\CE}

\date{}

\begin{document}

\newcommand\extrafootertext[1]{%
    \bgroup
    \renewcommand\thefootnote{\fnsymbol{footnote}}%
    \renewcommand\thempfootnote{\fnsymbol{mpfootnote}}%
    \footnotetext[0]{#1}%
    \egroup
}

\title{Lower Bounds for Subset Sum in Resolution with Modular Counting}

\author[1, 2]{Fedor Part\thanks{email: fedor.part@gmail.com}}

\affil[1]{Institute of Mathematics of the Czech Academy of Sciences}
\affil[2]{JetBrains Research}

\extrafootertext{\hspace{-18pt} The author was
supported by  the Institute of Mathematics, Czech Academy of Sciences (RVO~67985840)
and by GA\v{C}R grant 23-04825S.}

\maketitle

\begin{abstract}

In this paper we prove lower bounds for sizes of refutations of unsatisfiable vector Subset Sum
instances $\overline{a}_1 x_1 + \dots + \overline{a}_n x_n = \overline{b}$ in the proof
system $\reslin{\F_q}$, where $\F_q$ is a finite field of prime power order $q = q(n)$ with $char(\F_{q})\geq 5$. As a basis for the
hardness criterion for such instances we choose
the property of the matrix $A$ with columns $(\overline{a}_1, \ldots, \overline{a}_n)$ to
be (the transpose of) the generator matrix for a good error-correcting code 
$\mathcal C_{A} := \{x\cdot A\, |\, x \in \F_{q}^k\}\subset \F_{q}^n$ and prove the following lower bounds:
\begin{enumerate}
\item For a dag-like fragment of $\reslin{\F_q}$. We introduce the notion of $(s,r)$-robustness for Subset Sum instances, which in
particular implies that $A$ defines an error-correcting code with the minimal 
distance $s\geq r$. For $(s,r)$-robust instances we prove $2^{\Omega(r)}$ lower bound for sizes of refutations in a dag-like fragment of $\reslin{\F_q}$. We show that random
instances are $(n / 3, \Omega\left((n/(q + 1)\ln q))^{1/3}\right))$-robust and that specific examples achieving these bounds can be constructed using algebraic 
geometry codes.

\item For tree-like $\reslin{\F_q}$ refutations we show the size lower bound
$2^{\Omega({((q+1)\ln q)^{-1/3}}d^{1/5})}$ for any Subset Sum instance where $d$ is the minimal distance of
$\mathcal C_{A}$.
\end{enumerate}
\end{abstract}

\section{Introduction}

One of the central research programs in proof complexity, initiated by Cook and Reckhow \cite{CookReckhow79}, 
is to obtain the separation $\NP\neq \coNP$ via proving superpolynomial lower bounds for
sizes of proofs in all propositional proof systems. Among specific propositional proof systems that have been extensively studied are systems that operate with De Morgan formulas.
The most natural examples of such systems are \emph{Frege systems}, the standard ``textbook'' proof systems for propositional logic (for example, Hilbert-style or sequent calculus).
Obtaining superpolynomial lower bounds for general Frege systems turns out to be very hard, such bounds are probably out of reach in the nearest future. Nevertheless a significant progress
has been achieved: strong
lower bounds have been proven for AC$^0$-Frege systems, which are Frege systems restricted to operate only with bounded depth, unbounded fan-in De Morgan formulas \cite{Ajt88,PBI93,KPW95,BS02,Has17, Hastad23}.

Unfortunately all intricate lower bound techniques currently developed for proof systems operating with De Morgan formulas fail once De Morgan language is extended with
counting connectives. One of the long standing open problems at the frontier of proof complexity is to prove a superpolynomial lower bound for AC$^0[p]$-Frege systems, which
are defined similarly to AC$^0$-Frege systems with the difference that formulas can contain also counting modulo $p$ connective. This problem manifests itself already at the
level of resolution, the system operating with disjunctions of literals, that is with De Morgan formulas of depth $1$.

\emph{Resolution over linear equations} $\reslin{\mc R}$ is a propositional proof system extending resolution by allowing linear equations over a ring $\mathcal{R}$ in place of literals. 
Such extensions for various $\mathcal{R}$ equip resolution with the ability to perform basic counting arguments 
efficiently\footnote{Typically for $\mathcal{R} = \mathbb{Z}$ or for $\mathcal{R} = \mathbb{Z}_n$ for modular counting.}. Although $\reslin{\mc R}$ is a very simple extension of resolution,
which is the most well-understood propositional proof system with a plenty of exponential lower bounds proven for it, no superpolynomial lower bounds have been proven for unrestricted 
$\reslin{\mc R}$ to date. Thus $\reslin{\mc R}$ is a good starting point for devising novel techniques that work for proof systems combining De Morgan language with counting connectives since $\reslin{\mc R}$ is one of the simplest such systems and since 
$\reslin{\mc R}$ is very close to resolution, which is the most well-studied proof system. 

The strength of $\reslin{\mc R}$ needs to be understood primarily for 
$\coNP$-complete language of unsatisfiable CNFs, the central $\coNP$-complete language 
used to compare propositional proofs systems. However, in case $\mathcal{R} = \mathbb{Z}$
$\reslin{\mathbb{Z}}$ is also naturally a proof system for $\coNP$-complete language
$\mathsf{SubSum}$ of unsatisfiable subset sum instances $a_1x_1 + \dots + a_nx_n = b$, 
$x_i \in \{0, 1\}$. In \cite{PT18} Part and Tzameret
proved the first superpolynomial lower bound for dag-like $\reslin{\mathbb{Z}}$
considered as a proof system for $\mathsf{SubSum}$. This lower bound is a consequence of theorems proved in \cite{PT18}, which characterize $\reslin{\mathbb{Z}}$
as a proof system corresponding to \emph{dynamic programming}. It is not hard to see
that $a_1x_1 + \dots + a_nx_n = b$ admits dynamic programming refutations in dag-like $\reslin{\mathbb{Z}}$ which are of size 
$poly(|\{a_1{\nu}_1 + \dots + a_n{\nu}_n\, |\, \nu_i \in \{0,1\}\}|)$. One of the main theorems in \cite{PT18} is the lower bound 
$\Omega(|\{a_1{\nu}_1 + \dots + a_n{\nu}_n\, |\, \nu_i \in \{0,1\}\}|^{\epsilon})$,
for some $\epsilon < 1$, showing that $\reslin{\mathbb{Z}}$ cannot do better than dynamic programming. In particular,
$\reslin{\mathbb{Z}}$ refutations of the Binary Value Principle 
$x_1+2x_2+\dots+2^{n-1}x_n = -1$ are of size $2^{\Omega(n)}$.

In case $\mc R$ is a finite field $\F_q$ of size $q = poly(n)$ the subset sum problem
consisting of instances $a_1x_1 + \dots + a_nx_n = b$, $x_i \in \{0,1\}$, 
$a_i \in \F_{q(n)}$ is easy since it can be solved in polynomial time by dynamic
programming. If $char(\F_q) \geq 5$, a $\coNP$-complete analogue of $\mathsf{SubSum}$ 
in this setting is
the vector version of subset sum: unsatisfiable instances of the form
 $\overline a_1x_1 + \dots + \overline a_nx_n = \overline b$, 
$x_i \in \{0, 1\}$, $\overline{a}_1, \ldots, \overline{a}_n, \overline{b} \in \F_{q(n)}^k$
\footnote{We have to require at least that $q \geq 3$ since otherwise
the language is not $\coNP$-complete. The proof of $\coNP$-completeness in case
$char(\F_q) \geq 5$ can be found in \cite{PT18}.}.
This language, which we denote $\LinSys_{\F_q}$, comprises pairs $(A, b)$
defining linear systems $A\cdot x = b$ over $\F_{q(n)}$ without solutions in the boolean cube.

Aiming at 
constructing a bridge between lower bounds in \cite{PT18} and CNF lower bounds we initiate the development of methods for proving dag-like $\reslin{\F_q}$ lower bounds
for $\LinSys_{\F_q}$. Instances $A\cdot x = b$ are in some ways simpler than CNFs, this makes analysis of their $\reslin{\F_q}$ refutations more approachable. As we demonstrate, this analysis can benefit from techniques in linear algebra and additive combinatorics.

As in the case $\mc R = \mathbb{Z}$ proof systems $\reslin{\F_q}$ can refute 
instances $A\cdot x = b$ via dynamic programming with refutation of size polynomial
in the size $|A(\{0,1\}^n)|$ of the $A$-image of the boolean cube. However, in contrast
to $\reslin{\mathbb{Z}}$, this is not optimal: there exist
$A$, $b$ such that $|A(\{0,1\}^n)|$ is exponential, but $A\cdot x = b$ admits polynomial
size $\reslin{\F_q}$ refutations. For example, let $q\geq 3$, pick any $c\in \F_q\setminus\{0,1\}$,
and consider the $n\times n$ identity matrix $A = I_n$ with $b = (c,c,\ldots,c)$.
Then $A(\{0,1\}^n) = \{0,1\}^n$ has size $2^n$, yet for each $i$ the equation $x_i = c$
directly contradicts the boolean axiom $x_i = 0 \vee x_i = 1$ (since $c\notin\{0,1\}$),
yielding an $O(n)$-size $\reslin{\F_q}$ refutation. Proving lower bounds in this setting is significantly
more complicated since instances contain several equations.

Our ultimate goal is to identify natural hardness criterions for pairs $(A, b)$ simultaneously with discovering novel lower bound techniques allowing 
to prove lower bounds for these maps in dag-like $\reslin{\F}$. Especially valuable would be a hardness criterion for $(A, b)$, where $A$ is such that
the matrix $A$ contains small number of nonzero elements in each row, in that case $A\cdot x = b$ would have a short CNF encoding and
thus a lower bound for $A\cdot x = b$ would imply a lower bound for a CNF. 

Another interesting and related task is to identify natural hardness criterions for $(A, b)$ in polynomial calculus over $\F_q$, a pretty well-understood proof system. 
Firstly, polynomial calculus might serve as a testing ground for hardness for $\reslin{\F_q}$: polynomial calculus lower bounds imply tree-like $\reslin{\F_q}$ lower bounds \cite{PT18}.
But it is also of independent interest. In \cite{AlekhnovichRazborov01} Alekhnovich and Razborov formulated a polynomial calculus hardness criterion for 0-1 unsatisfiable systems 
$g_1 = 0, \ldots, g_m = 0$ of polynomial equations. Denote $M$ the 0-1 $m\times n$ matrix such that $g_i$ depends on $x_j$ iff $M_{ij} = 1$. Then the system $g_1 = 0, \ldots, g_m = 0$
is hard for polynomial calculus if each of $g_i$ is $l$-immune for $l > 1$ larger than
some small constant and $M$ is a good enough expander. However, the requirement for $g_i$ to be $l$-immune for $l > 1$ means that the characterisation of hard systems in \cite{AlekhnovichRazborov01} completely avoids hard 0-1 unsatisfiable systems of linear equations over $\F_q$.
Note that it is not hard to prove polynomial calculus lower bounds for some specific $(A, b)$, in fact
one can even define a class of $(A, b)$ constructed from CNFs and an implicit subclass, for which lower bounds follow from Alekhnovich-Razborov lower bounds, but this does not 
give rise to any \emph{natural} hardness criterion.

\subsection{Related work}

\subsubsection{Resolution over linear equations}

Proof system $\reslin{\mc R}$ was first introduced in the regime $\mathcal{R} = \mathbb{Z}$ by Raz and Tzameret in \cite{RT07}, where they showed that $\reslin{\mathbb{Z}}$
has short proofs for many instances that are frequently used as hard instances in proof complexity. Subsequently Itsykson and Sokolov introduced in \cite{IS14} proof system
$\reslin{\F_2}$,
proved several upper and lower bounds for \emph{tree-like} $\reslin{\F_2}$ and proved that $\reslin{\mathbb{Z}}$ p-simulates $\reslin{\F_2}$. 
In \cite{Khaniki21} Khaniki proved almost
quadratic lower bounds for resolution over polynomial equations over finite fields and thus, in particular, for $\reslin{\F_q}$. In \cite{GPT22} the computational model of \emph{linear branching programs} (LBPs) was
introduced, which is related to $\reslin{\F_2}$ in the same way as ordinary
boolean branching programs are related to resolution. More specifically, \cite{GPT22}
introduced the notions of \emph{strongly} and \emph{weakly read-once} LBPs
and showed that weakly read once LBPs correspond to what can naturally be called
the \emph{(weakly) regular} $\reslin{\F_2}$ whereas for strongly read-once LBPs
\cite{GPT22} proved exponential lower bounds. Subsequently \cite{EGI24} proved a
superpolynomial lower bound for refutations of the Bit Pigeonhole Principle in strongly regular $\reslin{\F_2}$. More recently, Alekseev and Itsykson \cite{AI25} proved exponential
lower bounds for bounded-depth dag-like $\reslin{\F_2}$ refutations of the Pigeonhole Principle; further progress on dag-like $\reslin{\F_2}$ lower bounds has been obtained in the subsequent follow-up works.  

\subsubsection{Binary value principle}

A variant of the Subset Sum Principle subsequently
called in \cite{AGHT19} the Binary Value Principle (BVP) is represented by the single equation $x_1+2x_2+\dots+2^{n-1}x_n = -1$, which is unsatisfiable over the Boolean 
assignments or, in other words, 0-1 unsatisfiable\footnote{Strictly speaking BVP is the negation of $x_1+2x_2+\dots+2^{n-1}x_n = -1$, it says that the natural number represented by the bit string $x_1\ldots x_n$ is never $-1$. 
$\reslin{\mc R}$ is like resolution a refutation system: it refutes $\neg\phi$ for a tautology $\phi$.}. 
As recent research in proof complexity showed, quite surprisingly, this simple principle turns out to be hard even for strong algebraic proof systems. 
In \cite{AGHT19} it was proved that BVP does not have short proofs even in the ideal proof system assuming Shub-Smale hypothesis. In \cite{ale21} Alekseev proved unconditional lower
bound for BVP in a pretty strong extension of polynomial calculus, where introduction of new variables and taking radicals are allowed.

\subsection{Our contributions}

Let $q$ be a prime power which can depend on $n$, the number of variables. We base hardness criterions for instances of the form $A\cdot x = b$ on the notions of an error correcting code and what we call $(s, r)$\emph{-robustness}, a combinatorial, algebraic 
property of linear systems $A\cdot x = b$, which we introduce. As a step towards general
$\reslin{\F_q}$ lower bounds for such instances we prove superpolynomial lower bounds
for a nontrivial dag-like fragment of $\reslin{\F_q}$ and also for two tree-like 
fragments, which capture features of $\reslin{\F_q}$ that are in a sense complementary to those captured by our dag-like fragment.

\begin{enumerate}
\item We consider a natural proof system $\LinDags_{\F_q}$ for $\LinSys_{\F_q}$ where
refutations are dags where nodes represent splittings on possible values of a linear form. 
This is a fragment of $\reslin{\F_q}$. Using dynamic programming argument one can show
that the fragment $\BinDags_{\F_q}$ of $\LinDags_{\F_q}$, which uses only binary splittings on variables, can simulate $\LinDags_{\F_q}$. We then consider a regular variant $\BinRegDags_{\F_q}$ of
$\BinDags_{\F_q}$ where each variable appears at most
once on every path from the root to a leaf. This proof system can no longer simulate
$\LinDags_{\F_q}$ since it can do dynamic argument at most once. Moreover, $\BinRegDags_{\F_q}$
is presumably weaker than regular $\reslin{\F_q}$ as defined, for instance, in 
\cite{EGI24}. However methods in \cite{EGI24} are not directly applicable in our setting for two reasons: we consider the case $q > 2$ whereas the results in \cite{EGI24} apply
only to $\F_2$, and we study linear systems whereas \cite{EGI24} study CNFs. 

 Still, $\BinRegDags_{\F_q}$
is nontrivial dag-like proof system which improves on dynamic programming and can efficiently refute systems like 
$\{x_i+2x_{n+i} = 0\}_{i \in [n]}\cup \{x_1 + \ldots + x_{2n} = 2\}$ (for this
system written as $A\cdot x = b$ the size of the image $|A(\{0,1\}^n)|$ is exponential,
which makes dynamic programming alone insufficient). Formal definitions of all proof systems are given in Section~\ref{sec:backgr}.

\item We prove $2^{\Omega(r)}$ lower bound for refutations in $\BinRegDags_{\F_q}$ of
instances $A\cdot x = b$ satisfying a property, which we call 
$(s, r)$\emph{-robustness}, $s\geq r$. In particular, for such instances the space
$\mc C_A := \{y\cdot A\, |\, y \in \F_q^k\}$ is an error correcting code (ECC) with
the distance at least $s$.  The dag-like lower bound is in Section~\ref{bindagLB}.

\item We construct $(s,r)$-robust instances from random linear codes and
algebraic geometry codes (including Hermitian codes), giving semi-explicit hard
families. We show that random instances, obtained by randomly choosing
$A$, are with high probability $(n / 3, \Omega\left((n/(q + 1)\ln q))^{1/3}\right))$-robust. As a consequence, we obtain that all explicit codes that meet the Gilbert-Varshamov bound for random codes, such as algebraic geometry codes, can be used to construct
explicit instances with these parameters (at the moment we can only choose $A$). We raise several open problems: the exact threshold $\Delta(k,q)$ for
0-1 satisfiability, existence of robust instances for small~$s$, and extensions
of our criterion to Nullstellensatz and Polynomial Calculus.
Sections~\ref{hardInstSec}--\ref{sec:robust} and the Conclusion.

\item We prove $2^{\Omega(((q + 1)\ln q)^{-1/3}d^{1/5})}$ lower bounds for
tree-like $\reslin{\F_q}$ and $\LinTrees_{\F_q}$ (the tree-like version of
$\LinDags_{\F_q}$) for every instance $A\cdot x = b$ where $A$ is a generator
matrix for an ECC with minimal distance $d$.  The proof uses a Prover-Delayer
game: we exhibit a Delayer strategy guaranteeing
$\Omega(((q+1)\ln q)^{-1/3}d^{1/5})$ branching points.
Sections~\ref{sec:tree-like} and~\ref{bindagLB}.

\item At the heart of both lower bounds is a characterisation of when
$\ECC^{n,k,d}_{\F_q}$ is non-empty.  We show: (a)~if $n < (\log_2 q)\,k$ then
hard instances exist; (b)~if $d \geq ((q+1)\ln q)\,k^3$ then every system with
$d_A \geq d$ is 0-1 satisfiable (no hard instances with such large distance
exist).  The proof of~(b) reduces to a Minkowski-sum lemma in additive
combinatorics: any $t \geq ((q+1)\ln q)\,k^2$ bases of $\F_q^k$ satisfy
$X_1+_M\cdots+_MX_t = \F_q^k$.  Section~\ref{hardInstSec}.

\end{enumerate}

\medskip\noindent\textbf{Organisation.}
Section~\ref{sec:backgr} defines $\reslin{\F_q}$, $\reslinneq{\F_q}$, and their
fragments ($\LinDags_{\F_q}$, $\BinDags_{\F_q}$, $\BinRegDags_{\F_q}$,
$\LinTrees_{\F_q}$), together with the Prover-Delayer games.  Notation is in
Section~\ref{sec:notation}.  Section~\ref{hardInstSec} develops ECC-based hard
instances: emptiness/non-emptiness criteria for $\ECC^{n,k,d}_{\F_q}$,
$(s,r)$-robustness, and its verification for random and Hermitian codes.
Sections~\ref{sec:tree-like} and~\ref{bindagLB} contain the lower-bound proofs.
Open problems are in the Conclusion and in Section~\ref{sec:robust}.

\tableofcontents

\section{Background and notation}\label{sec:backgr}  

\subsection{Proof system $\reslin{\F_q}$}

Proof lines of $\reslin{\F_q}$ are disjunctions of linear equations over a $\F_q$: 
$\left(\sum\nolimits_{i=0}^na_{1i}x_i+b_1=0\right)\vee\dots\vee\left(\sum\nolimits_{i=0}^na_{ki}x_i+b_k=0\right)$. 
The rules of $\reslin{\F_q}$ are as follows (cf.~\cite{RT07}): 

\begin{prooftree}
        \centering
        \def\labelSpacing{12pt}
        \AxiomC{$C \vee f=0$}
        \AxiomC{$D \vee g=0$}
        \RightLabel{($\alpha,\beta\in \F_q$)}
        \LeftLabel{(Resolution)}
        \BinaryInfC{$C \vee D \vee \left(\alpha f+ \beta g\right)=0$}
\end{prooftree}

\begin{prooftree}
        \AxiomC{$C \vee a=0$}
        \RightLabel{($0\neq a\in \F_q$)}
        \LeftLabel{(Simplification)}
        \UnaryInfC{$C$}
        
        \AxiomC{$C$}
        \LeftLabel{(Weakening)}
        \UnaryInfC{$C \vee f=0$}
        \noLine
        \BinaryInfC{}
\end{prooftree}
where $f,g$ are linear polynomials over $\F_q$ and $C, D$ are \emph{linear clauses}.
A \emph{clause} is a disjunction of literals (positive or negative boolean variables),
and a \emph{linear clause} is a disjunction of linear equations $f = a$ or inequalities
$f \neq a$ (in case of $\reslinneq{\F_q}$ which we define later) over $\F_q$.
We use the term ``clause'' broadly for disjunctions of literals, of linear equations,
or of linear inequalities, as appropriate from context.
The \emph{boolean axioms} are defined as follows: 
$$ x_i=0 \lor x_i=1 \text{, ~for $x_i$ a variable}$$
A $\reslin{\F_q}$ \emph{derivation} of a linear clause $D$ from a set of linear clauses $\phi$ is a
sequence of linear clauses $(D_1,\dots,D_s\equiv D)$ such that for every $1\leq i\leq s$ 
either $D_i\in \phi$ or is a boolean  axiom or $D_i$ is obtained from previous clauses by applying one of the rules above and $D_s\equiv D$ means $D$ coincides with $D_s$.
A $\reslin{\F_q}$ \emph{refutation} of an unsatisfiable
set of linear clauses $\phi$ is a $\reslin{\F_q}$ derivation of the empty clause from $\phi$. 

\subsection{Notation and conventions}\label{sec:notation}

A \emph{linear error-correcting code} (ECC) $\mc C \subseteq \F^n$ of dimension $k$ is a $k$-dimensional subspace.
Its \emph{generator matrix} $A$ is any $k\times n$ matrix whose rows form a basis of $\mc C$,
so that $\mc C = \{x\cdot A \mid x\in \F^k\}$\footnote{Our definition slightly differs from the standard one in the literature: our generator matrix corresponds to the standard generator matrix transposed.}. Given $A$ we denote the corresponding code
$\mc C_A$. Generator matrices occurring throughout the paper are matrices $A$ in equations $A\cdot x = b$.

For a vector $v \in \F^n$ we define its \emph{weight} $\omega(v) := |\{i \mid v_i \neq 0\}|$ as the number of its nonzero coordinates. For a subspace $V \subseteq \F^n$ we set $\omega(V) := \min_{v\in V,\, v\neq 0}\omega(v)$. A linear subspace $\mc C \subseteq \F^n$ is an ECC with parameters $(n, k, d)$ if $\dim(\mc C) = k$ and $\omega(\mc C) = d$. If $A$ is a 
$k \times n$ matrix over $\F$, we also use the notation $d_A := \omega(\mc C_A)$ for the minimal (Hamming) distance of the code $\mc C_A$.

For a set $\mathcal{F}$ of vectors in a vector space, we denote by $\langle \mathcal{F} \rangle$ the vector subspace generated by $\mathcal{F}$. We define the \emph{weight-truncated span} $[V]_{\omega \leq \tau} := \langle \{v \in V \mid \omega(v) \leq \tau\} \rangle$ as the vector subspace generated by all vectors in $V$ of weight at most $\tau$.

We denote $[n] := \{1, \ldots, n\}$.  For subsets $V, W \subseteq \F^n$ we write $V + W := \langle V\cup W\rangle$ for the linear span of their union, and $V +_M W := \{v + w \mid v \in \langle V\rangle,\, w \in \langle W\rangle\}$ for their \emph{Minkowski sum}.  Note that $V+_M W \subseteq V+W$, but $V+_M W$ is not necessarily a linear subspace and in general is a proper subset of $V+W$.  If $V \subseteq \F^n$ and $I \subseteq [n]$, we write $V_{I}$ or $V_{[I]}$ for the set of vectors in $\F^{|I|}$ obtained from vectors of $V$ by retaining only the coordinates indexed by $I$.

If $\mathcal{F}$ is a set of linear polynomials with $n$ variables (the set is unordered, but once an implicit ordering of the rows of the corresponding matrix is fixed), we write $\mathcal{F}_{[1]}$ for the \emph{tuple} of linear forms (the degree-$1$ parts) in $\mathcal{F}$, and $\mathcal{F}_{[0]}$ for the \emph{tuple} of their free coefficients (degree-$0$ terms).  The use of ``tuple'' (rather than ``set'') reflects the fact that these objects are ordered consistently with the rows of the coefficient matrix.  For readability we will sometimes interpret a vector space of linear polynomials $P$ as a vector space of linear equations, writing both $f - a \in P$ and $f = a \in P$ for $a\in\F$.

For a vector space $P$, we let $red_P(h)$ be any $h'\in P$ such that $h' \equiv \alpha h + f$ for some $f \in P$, $\alpha\neq 0 \in \F$, and $h'$ has minimal weight among all such vectors.\footnote{Here $\equiv$ means ``coincides as a polynomial''. Since we are often dealing with spaces of equations we use $\equiv$ for equality of their elements to avoid confusion with ``$=$'' used inside equations.}  Finally, for linear spaces of equations $\mc F$ and $\mc G$ we write $\mc F\models_{0,1} \mc G$ if every equation $g = a \in \mc G$ is implied by $\mc F$ over 0-1 assignments.

\subsection{Random codes and algebraic geometry codes}\label{sec:RScode}

\begin{theorem}[Gilbert-Varshamov bound, see e.g.\ \cite{AGCbook}]\label{thm:GV}
Let $q \geq 2$ be a prime power and let $A$ be a uniformly
random $k \times n$ matrix over $\F_q$.  Consider the linear code
$\mc C_A = \{x\cdot A\, |\, x\in \F_q^k\}$ with generator matrix $A$; it has
parameters $(n, k, d)$, where $n$ is the length, $k = \dim(\mc C_A)$ is the dimension,
and $d = \omega(\mc C_A)$ is the \emph{minimal (Hamming) distance}, i.e.\ the smallest
Hamming weight of a nonzero codeword.
Let $H_q(x) := x\log_q (q - 1) - x\log_q x - (1-x)\log_q (1-x)$ be $q$-ary entropy and
denote $R := k / n$, $\delta := d / n$.  If $0 \leq \delta_0 < 1 - 1 / q$ then
the probability that
$d / n \geq \delta_0$ is at least $1 - q^{-\eta n}$ where
$\eta := (1 - H_q(\delta_0) - R)$.
\end{theorem}

The proof of the theorem is straightforward and is folklore. For our hardness results on random instances, we will need the following corollary, which gives an explicit bound on minimal distance for random codes:

\begin{corollary}[Random codes]\label{cor:rand_code_dist}
Let $q > 8$ be a prime power and let $A$ be a uniformly random $k \times n$ matrix over
$\F_q$ where $k \geq \lfloor n / \log q \rfloor$. Then
with high probability the $(n, k, d)$-code $\{x\cdot A\, |\, x\in \F_q^k\}$ with generator matrix $A$ has distance $d \geq n - 3k$.
\end{corollary}
\begin{proof}
We apply Theorem~\ref{thm:GV} with $\delta_0 := 1 - 3R$ where $R := k/n$.
We must verify: (i) $\delta_0 \in [0, 1-1/q)$, and (ii) $\eta := 1 - H_q(\delta_0) - R > 0$.

\medskip
\noindent\textit{Step 1: check $\delta_0 \geq 0$.}
We have $\delta_0 = 1 - 3k/n \geq 0$ iff $k \leq n/3$.  Since $k \geq n/\log q$ and
$q > 8$ we have $k/n \leq 1/3$ as required.

\medskip
\noindent\textit{Step 2: check $\delta_0 < 1-1/q$.}
We need $3R > 1/q$, i.e.\ $k > n/(3q)$.  Since $k \geq n/\log q$ and $\log q < 3q$
for all $q\geq 2$, this holds.

\medskip
\noindent\textit{Step 3: lower bound $\eta$.}
Using $H_q(\delta_0) \leq (1-3R)\log_q(q-1) + H_2(\delta_0)/\log q$ and
$H_2(\delta_0) \leq 1$:
$$\eta \geq 1 - (1 - 3R)\log_q(q-1) - \frac{1}{\log q} - R.$$
Since $\log_q(q-1) < 1$ and $q > 8$ implies $1/\log q < 1/3 < R$ we get
$\eta > 1 - (1-3R) - R - R = 2R - R = R > 0$.
(A more careful calculation gives $\eta > 2R - 1/\log q \geq R > 0$.)

\medskip
By Theorem~\ref{thm:GV}, the probability that $d/n \geq \delta_0 = 1 - 3R$ fails
is at most $q^{-\eta n} \leq q^{-Rn}$, which goes to $0$ with $n$.
\end{proof}

We can construct explicit hard instances out of any code with the distance satisfying
the bound for random codes from Corollary~\ref{cor:rand_code_dist}. A wide range of
examples of such codes is provided by algebraic geometry codes\cite{AGCbook}. We pick Hermite codes.

Let $q := p^2$ for a prime $p$ and consider the Hermitian curve
$\mc H_q := \{(x,y)\in \F_q^2\, |\, x^p + x = y^{p+1}\}$.
The number of $\F_q$-rational points on $\mc H_q$ is $p^3$ which can be shown as 
follows.

Let $y \in \mathbb{F}_q$ be chosen arbitrarily. Since $q = p^2$, there are exactly $p^2$ such choices for $y$.

For any $y \in \mathbb{F}_q$, the expression $y^{p+1}$ is the relative norm $N_{\mathbb{F}_q/\mathbb{F}_p}(y)$. By the fundamental properties of the norm map, the image of $N_{\mathbb{F}_q/\mathbb{F}_p}$ is contained entirely within the base field $\mathbb{F}_p$. Let us denote this resulting value by $c$, so we have $c \in \mathbb{F}_p$.

For this fixed $y$ (and therefore fixed $c$), we must find the number of corresponding solutions $x \in \mathbb{F}_q$ to the curve's equation:
$$x^p + x = c$$

Notice that the left-hand side of this equation is exactly the field trace $\mathrm{Tr}_{\mathbb{F}_q/\mathbb{F}_p}(x)$. The trace map $\mathrm{Tr}_{\mathbb{F}_q/\mathbb{F}_p}: \mathbb{F}_q \to \mathbb{F}_p$ is an $\mathbb{F}_p$-linear transformation from $\mathbb{F}_q$ (which is a $2$-dimensional vector space over $\mathbb{F}_p$) to $\mathbb{F}_p$ (a $1$-dimensional vector space over $\mathbb{F}_p$).

The trace map is surjective, meaning its image has dimension $1$. By the Rank-Nullity Theorem, the dimension of its kernel is $\dim_{\mathbb{F}_p}(\ker(\mathrm{Tr})) = 2 - 1 = 1$. Therefore, the equation $x^p + x = c$ yields exactly $p$ solutions for $x \in \mathbb{F}_q$, regardless of the specific value of $c$.

Since there are $p^2$ independent choices for $y$, and every choice of $y$ yields exactly $p$ corresponding solutions for $x$, the total number of $\mathbb{F}_q$-rational points on the Hermitian curve $\mathcal{H}_q$ is (see e.g.\ \cite{AGCbook}):
$$p^2 \cdot p = p^3$$

Denote these points $P_1,\dots, P_{p^3}$. Define the Hermitian code of
degree $D$ as: 
$$\mc C_D := \{(R(P_1), \ldots, R(P_{p^3}))\, |\, R(x, y) \in \F_q[x,y], deg(R) = D, deg_y(R)\leq p\}$$

The proposition below characterizes the code parameters for $D := \lfloor p^3 / \log_2 p \rfloor$.

\begin{proposition}[Hermitian codes, \cite{AGCbook}]\label{prop:hermit_code_dist}
Let $D := \lfloor p^3 / \log_2 p \rfloor$
for a prime $p$. The
parameters of the Hermitian code $\mc C_D$ over $\F_{p^2}$ are as follows:
$n = p^3$,
$k = \dim(\mc C_D) \geq \lfloor p^3/\log_2 p - p^2/2 \rfloor$,
$d \geq \left(1 - 1/\log_2 p\right) n$.
\end{proposition}
\begin{proof}[Proof Sketch]
To understand this code, it helps to think of it as a generalized Reed-Solomon code. Instead of evaluating standard polynomials on a straight line, we evaluate generalized ``polynomial-like'' functions on the Hermitian curve $\mathcal{H}_p$ over the field $\mathbb{F}_{p^2}$.

From the geometry of the curve, we know two key facts:
1. The curve has $n = p^3$ affine points, denoted $\mathcal{P} = \{P_1, \dots, P_n\}$. These will serve as evaluation points.
2. The curve has a topological ``complexity'' parameter called the genus, $g = p(p-1)/2$. Geometrically, the genus represents the number of ``holes'' in the curve's surface. In coding theory, this acts as a penalty: higher genus means a more twisted geometry, which restricts how many independent functions we can build.

We construct the code $\mathcal{C}_D$ by taking a vector space of functions and evaluating them at all $n$ points. To control the properties of the code, we restrict our functions so that they only blow up (have poles) at a single point not in our evaluation set: the ``point at infinity'' $P_\infty$. We cap the maximum allowed degree of these poles at $D := \lfloor p^3 / \log_2 p \rfloor$. 

For the mathematical machinery to work perfectly, this degree parameter $D$ must sit in a ``Goldilocks zone'' relative to the genus and the number of points: $2g - 2 < D < n$. For primes $p \geq 3$, a quick check confirms that our choice of $D$ safely satisfies these bounds.

\textbf{Minimum Distance ($d$):} 
Just like a polynomial of degree $D$ can have at most $D$ roots, a function in our space can evaluate to zero at most $D$ times across our curve. Therefore, two distinct codewords can agree on at most $D$ coordinates. This gives a strict lower bound on the minimum distance:
$$d \geq n - D = p^3 - \lfloor \frac{p^3}{\log_2 p} \rfloor \geq p^3\left(1 - \frac{1}{\log_2 p}\right)$$

\textbf{Dimension ($k$):} 
To find the dimension of the code, we need to know how many independent functions exist in our space. Because $D > 2g - 2$, a fundamental result called the Riemann-Roch theorem acts as a clean dimension-counting formula, guaranteeing that the number of independent functions is exactly:
$$k = D - g + 1 = \lfloor \frac{p^3}{\log_2 p} \rfloor - \frac{p^2}{2} + \frac{p}{2} + 1$$
By dropping the strictly positive lower-order terms ($\frac{p}{2} + 1$), we arrive at a lower bound for the dimension: $k \geq \lfloor p^3/\log_2 p - p^2/2 \rfloor$, which completes the proof.
\end{proof}

\section{Proof system $\reslinneq{\F_q}$ and its fragments}

Throughout this section, $q = q(n)$ denotes a prime power that may grow with $n$
(the number of variables); we assume $q(n) = \mathrm{poly}(n)$ unless stated otherwise,
so that $|\F_q| = q$ is polynomially bounded.

We define a proof system $\reslinneq{\F_q}$, which is in a sense a dual version of $\reslin{\F_q}$. Its proof lines are disjunctions of linear \emph{inequalities} over a $\F_q$: 
$\left(\sum\nolimits_{i=0}^na_{1i}x_i+b_1\neq 0\right)\vee\dots\vee\left(\sum\nolimits_{i=0}^na_{ki}x_i+b_k \neq 0\right)$. 
The rules of $\reslinneq{\F_q}$ are as follows: 

\begin{prooftree}
        \centering
        \def\labelSpacing{12pt}
        \AxiomC{$C_1 \vee f \neq 0$}
        \AxiomC{$\dots$}
        \AxiomC{$C_{q-1} \vee f\neq q - 1$}
        \LeftLabel{(Resolution)}
        \TrinaryInfC{$C_1\vee \dots\vee C_{q-1}$}

\end{prooftree}

\begin{prooftree}
        \AxiomC{$C \vee 0\neq 0$}
        \LeftLabel{(Simplification)}
        \UnaryInfC{$C$}
        
        \AxiomC{$C \vee f \neq a$}
        \RightLabel{($a, b \in \F_q$)}
        \LeftLabel{(Linear combination)}
        \UnaryInfC{$C \vee f + g \neq a + b \vee g \neq b$}
        \noLine
        \BinaryInfC{}
\end{prooftree}
where $f,g$ are linear polynomials over $\F_q$ and $C, C_1,\ldots, C_{q-1}$ are clauses. The boolean axioms in this case are defined as follows: 
$$ x_i\neq 2, \ldots,  x_i\neq q - 1 \text{, ~for $x_i$ a variable}$$

\begin{proposition}
$\reslin{\F_q}$ p-simulates $\reslinneq{\F_q}$ if $q(n) = O(1)$.
\end{proposition}
\begin{proof}
$\reslin{\F_q}$ can encode an inequality $f\neq a$ as disjunction $[[ f\neq a ]] := \bigvee\limits_{b \in f(\{0,1\}^n), b\neq a}f = b$. For a $\reslinneq{\F_q}$ clause 
$C = f_1 \neq a_1 \vee \dots \vee f_m \neq a_m$ denote $[[ C]] := [[ f_1 \neq a_1]] \vee \dots \vee [[ f_m \neq a_m ]]$ the $\reslin{\F_q}$ encoding of $C$.
We just need to show that $\reslin{\F_q}$ simulates $\reslinneq{\F_q}$ resolution and linear combination rules.

In \cite{PT18} it was shown that $\reslin{\F_q}$ has polynomial size refutations of the \emph{image avoidance principle}: 
$\mathsf{ImAv}(f) := \{[[ f\neq a]]\}_{a \in f(\{0,1\}^n)}$.

\begin{proposition*}[Proposition 30 in \cite{PT18}]
Let $\mc R$ be a finite ring, $f = a_1x_1 + \dots + a_n x_n$ a linear form over $\mc R$.
Denote $\mathsf{im}(f) := f(\{0,1\}^n) = \{f(\nu) : \nu\in\{0,1\}^n\}$ the image of $f$
on boolean assignments,
$s_f$ the total size of the binary encoding of $\mathsf{im}(f)$ as a set of field elements, and
$d_f := |\mathsf{im}(f)|$ the number of distinct values. Then there exists a tree-like
$\reslin{\mc R}$ refutation of $\mathsf{ImAv}(f) := \{[[f\neq a]]\}_{a \in \mathsf{im}(f)}$
of size $O(s_fn^{2d_f})$.
\end{proposition*}

Another fact from \cite{PT18} that we need: short derivations of 
$\mathsf{Im}(f) := \bigvee_{c \in f(\{0,1\}^n)}f = c$

\begin{proposition*}[Proposition 12 in \cite{PT18}]
Let $f = a_1x_1 + \dots + a_n x_n + b$ be a linear polynomial over a ring $\mc R$. There exists a $\reslin{\mc R}$
derivation of $\mathsf{Im}(f)$ of size polynomial in $|\mathsf{im}(f)|$.
\end{proposition*}

Weakening clauses $\{[[ C_a ]] \vee [[ f\neq a]]\}_{a \in f(\{0,1\}^n)}$ to 
$\{ \bigvee_{b \in f(\{0,1\}^n)}[[ C_b ]]  \vee [[ f\neq a]]\}_{a \in f(\{0,1\}^n)}$ and applying refutation of $\mathsf{ImAv}$($f$) (the one from Proposition 30 in \cite{PT18} restated above) weakened with 
$\bigvee_{a \in f(\{0,1\}^n)}[[ C_a ]]$ we derive $\bigvee_{a \in f(\{0,1\}^n)}[[ C_a ]]$. This simulates resolution rule.

Consider a clause $[[ C ]] \vee [[ f \neq a ]]$. By Proposition~12 of \cite{PT18}
(restated above) there exists a $\reslin{\F_q}$ derivation of
${\bigvee_{c \in g(\{0,1\}^n), c\neq b}g = c}\,\vee\, g = b$
of polynomial size. By resolving each $f = c$ in $[[ f \neq a ]]$ with $g = b$ in ${\bigvee_{c \in g(\{0,1\}^n), c\neq b}g = c}\, \vee\, {g = b}$ we derive 
$[[ C ]] \vee [[ f + \alpha g \neq a + \alpha b ]] \vee [[ g\neq b ]]$. This simulates linear combination rule.

\end{proof}

Henceforth, all proof systems are considered as proof systems for the \coNP-complete
language $\LinSys_{\mc R}$ of 0-1 unsatisfiable linear systems $A\cdot x = b$ over
a ring $\mc R$.

\begin{definition}[$\LinSys_{\mc R}$]
Let $\mc R$ be a ring with $char(\mc R)\notin \{2,3,4\}$. $\LinSys_{\mc R}$ is \coNP\ language of pairs $(A, b)$ such that $A\cdot x = b$ is a 0-1 unsatisfiable linear system over $\mc R$.
\end{definition}

\begin{proposition}[\cite{PT18}]
$\LinSys_{\mc R}$ is \coNP-complete.
\end{proposition}

\begin{remark}
The characteristic restriction $char(\mc R)\notin\{2,3,4\}$ in the definition above is
from \cite{PT18}. Gryaznov~\cite{Gryaznov19} subsequently proved \coNP-completeness
of the analogous language for all characteristics other than $2$ and $3$
(i.e.\ removing the restriction $char\neq 4$). The results of the present paper apply
for $char(\F_q)\geq 5$, which is consistent with both references.
\end{remark}

We now define a proof system $\LinDags_{\F_q}$, which is p-equivalent to $\reslinneq{\F_q}$ and is just a more convenient way to represent $\reslinneq{\F_q}$ proofs for 
$\LinSys_{\F_q}$.

\begin{definition}
Let $\mathcal{F}=\{l_1(x_1,\ldots,x_n)=0, \ldots, l_m (x_1,\ldots,x_n)= 0\}$ be a linear system over $\F_q$  without solutions in the boolean cube $\{0,1\}^n\subset \F_q^n$.
A $\LinDags_{\F_q}$ refutation $T$ of $\mathcal{F}$ is a dag such that:
\begin{itemize}
\item Every node $v\in T$ is marked with a 0-1 unsatisfiable system $\mathcal{F}_v$. If $v$ has outgoing degree $0$ it is a \textbf{terminal node}, otherwise it is a \textbf{splitting node},
it is marked with a linear form $f_v$ and has outgoing degree $|f_v(\{0,1\}^n)|$. There is exactly one node with ingoing degree $0$ (\textbf{the root}). 
\item If $r \in T$ is the root then $\mathcal{F}_r = \mathcal{F}$.
\item If $v \in T$ is a splitting node then outgoing edges from $v$ lead to nodes $\{v_a\}_{a \in f_v(\{0,1\}^n)}$ and an edge $(v, v_a)$ is marked with equation $f_v = a$. 
The following condition must hold: 
$$\langle\mathcal{F}_{v_a}\rangle\subseteq \langle\mathcal{F}_v, f_v = a\rangle, \quad a\in f_v(\{0,1\}^n)$$
\item If $v \in T$ is a leaf then $\mathcal{F}_v$ has no solutions in $\F_q^n$.
\end{itemize}
\end{definition}

\begin{definition}\label{binDagsDef}
The proof system $\BinDags_{\F_q}$ is obtained by restricting $\LinDags_{\F_q}$ proofs to splittings on variables, that is $\BinDags_{\F_q}$ proofs are
precisely those $\LinDags_{\F_q}$ proofs, where  all splitting nodes are marked with variables.
\end{definition}

The following simulation holds:

\begin{proposition}
The proof system $\BinDags_{\F_q}$ is p-equivalent to $\LinDags_{\F_q}$.
\end{proposition}
\begin{proof}
$\BinDags_{\F_q}$ p-simulates 
$\LinDags_{\F_q}$ (and thus is p-equivalent)
via dynamic programming: a branching on a linear form 
$f := a_1x_1 +\ldots +a_nx_n$ an be simulated by sequentially branching on 
$f_i := a_1x_1 +\ldots +a_ix_i$ inductively on $i$ using binary splittings at each step.
\end{proof}

\begin{definition}
Let $\mathcal{F}=\{l_1(x_1,\ldots,x_n)=0, \ldots, l_m (x_1,\ldots,x_n)= 0\}$ be a system of linear equations over $\F_p$  without solutions in the boolean cube $\{0,1\}^n\subset \F_q^n$.
A $\BinRegDags_{\F_q}$ refutation $T$ of $\mathcal{F}$ is a dag such that:
\begin{itemize}
\item Every node $v\in T$ is marked with a 0-1 unsatisfiable system $\mathcal{F}_v$. If $v$ has outgoing degree $0$ it is a \textbf{terminal node}, otherwise it is a \textbf{splitting node},
it is marked with a variable $x_v$ and has outgoing degree $2$. There is exactly one node with ingoing degree $0$ (\textbf{the root}). 
\item If $r \in T$ is the root then $\mathcal{F}_r = \mathcal{F}$.
\item If $v \in T$ is a splitting node then outgoing edges from $v$ lead to nodes $v_0, v_1$ and an edge $(v, v_b)$ is marked with equation $x_v = b$, $b\in \{0,1\}$. 
The following condition must hold: 
$$\langle\mathcal{F}_{v_b}\rangle\subseteq \langle\mathcal{F}_v\rangle\rst_{x_v\leftarrow b}, b\in \{0,1\}$$
\item If $v \in T$ is a leaf then $\mathcal{F}_v$ has no solutions in $\F_q^n$.
\end{itemize}
\end{definition}

\begin{remark}[Capabilities of $\BinRegDags_{\F_q}$]\label{rem:brd-cap}
Although $\BinRegDags_{\F_q}$ enforces a strong regularity condition,
it strictly extends dynamic programming in two ways.

\emph{(1)~Dynamic programming.}
For every $(A,b)\in\LinSys_{\F_q}$ there is a $\BinRegDags_{\F_q}$ refutation of
size $O(n\cdot|A(\{0,1\}^n)|)$, obtained by connecting the layers
$L_i:=\{(A\cdot x)\rst_{x_1\leftarrow a_1,\ldots,x_i\leftarrow a_i}=b
\mid(a_1,\ldots,a_i)\in\{0,1\}^i\}$.

\emph{(2)~Dynamic programming with subsystem reduction.}
A smart dynamic-programming strategy may use any derived system
$\mathcal{G}\subseteq\langle\mathcal{F}\rangle$, not only~$A$ itself.
By Gaussian elimination on the first $\varepsilon n$ columns, assume
$\mathcal{G}$ has \emph{block form}
\[
  \mathcal{G} \;=\;
  \begin{pmatrix} I & D \\ C & B \end{pmatrix},
\]
where $I$ is the $\varepsilon n\times\varepsilon n$ identity block,
$D$ fills the remaining $n{-}\varepsilon n$ columns in those same rows,
$C$ is the $(m{-}\varepsilon n)\times\varepsilon n$ block below~$I$
with $\mathrm{rank}(C)=r$, and $B$ is the
$(m{-}\varepsilon n)\times(n{-}\varepsilon n)$ bottom-right block.

We now give an example of how a system $\mc G$ in principle may collapse step by
step losing equations after substitutions.
Consider branching on $x_1,\ldots,x_{\varepsilon n}$ in sequence, forming the
\emph{ladder} of Figure~\ref{fig:ladder}.  The two cases at each rung
behave very differently:
\begin{itemize}
\item \textbf{$x_i\leftarrow 0$.}  Column~$i$ is removed but
  \emph{no equation is discarded}; the rectangle loses only one column.
\item \textbf{$x_i\leftarrow 1$.}  We assume that setting
  $x_i\leftarrow 1$ allows rows~$1,\ldots,i$ to be
  \emph{discarded}(given all previous variable are set)---we assume the subsystem formed by remaining equations
  becomes 0-1-unsatisfiable.
  The rectangle then drops to $m{-}i$ rows and $n{-}i$ columns.
  However, the right-hand side of the remaining rows still depends
  on the partial assignment~$\rho=(x_1,\ldots,x_i)$ through the
  block~$C$: since $\mathrm{rank}(C)=r$, the map $\rho\mapsto C\rho$
  takes at most~$2^r$ distinct values, so every $x_i{=}1$ arc from
  rung~$i{-}1$ lands on \textbf{one of at most~$2^r$ nodes} in spine
  with $m{-}i$ equations.
\end{itemize}

The reachable (rung,~equation-count) pairs form a lower-left triangle:
position $(j,\,m{-}k)$ with $0\le k\le j\le\varepsilon n$, giving
$O((\varepsilon n)^2)$ distinct positions.  At each position the
coefficient matrix is determined, but the right-hand side depends on
the partial assignment~$\rho$ through the block~$C$: since
$\mathrm{rank}(C)=r$, the map $\rho\mapsto C\rho$ takes at most~$2^r$
distinct values on~$\{0,1\}^{\varepsilon n}$.  Hence \emph{every
position hosts at most $2^r$ systems} (same matrix, different RHS),
and the total DAG size is
\[
  O\!\left((\varepsilon n)^{2}\cdot 2^{r}\right),
\]
far below the na\"{\i}ve $2^{\varepsilon n}$.  With $r=O(\log n)$
this is polynomial; with $r=O(1)$ it is \emph{quadratic} in
$\varepsilon n$. 

Our $(s,r)$-robustness condition (Definition~\ref{def:robust}) is
designed precisely to close this way of refuting: it requires that
\emph{every} subsystem $\mathcal{G}\subseteq\langle\mathcal{F}\rangle$
whose restriction to $\mathrm{supp}(\rho)$ is 0-1-unsatisfiable must
have rank at least~$r$ on $\mathrm{supp}(\rho)$.  Consequently, no
derived system of rank less than~$r$ can serve as a witness, and any
dynamic-programming-with-reduction strategy requires at least~$2^r$
sub-problems and cannot yield a short refutation.
\end{remark}

\newcommand{\drawmatrix}[8][M]{%
  \begin{scope}[shift={(#2,#3)}]
    \pgfmathsetmacro{\cellsz}{1.6}%
    \pgfmathsetmacro{\W}{#5*\cellsz/10}%
    \pgfmathsetmacro{\H}{#4*\cellsz/10}%
    \pgfmathsetmacro{\idW}{#7*\cellsz/10}%
    \pgfmathsetmacro{\idH}{#6*\cellsz/10}%
    \pgfmathsetmacro{\botH}{\H-\idH}%
    \pgfmathsetmacro{\rightW}{\W-\idW}%
    \fill[#8, opacity=0.20] (0.10,0.10) rectangle ++(\W,\H);
    \fill[#8, opacity=0.35] (0.05,0.05) rectangle ++(\W,\H);
    \fill[green!12] (\idW,0) rectangle (\W,{\H-\idH});
    \pgfmathtruncatemacro{\idN}{#6}%
    \ifnum\idN>0
      \fill[gray!8] (\idW,{\H-\idH}) rectangle (\W,\H);
    \fi
    \ifdim\idW pt>0.05pt
      \ifdim\botH pt>0.05pt
        \fill[red!12] (0,0) rectangle (\idW,{\H-\idH});
      \fi
    \fi
    \ifnum\idN>0
      \fill[blue!25] (0,{\H-\idH}) rectangle (\idW,\H);
      \pgfmathsetmacro{\step}{\idH/\idN}%
    \fi
    \draw (0,0) rectangle (\W,\H);
    \ifdim\idH pt>0.05pt
      \draw[thin, gray] (0,{\H-\idH}) -- (\W,{\H-\idH});
    \fi
    \ifdim\idW pt>0.05pt
      \draw[thin, gray] (\idW,0) -- (\idW,\H);
    \fi
    \coordinate (#1) at (\W/2, \H/2);
    \coordinate (#1-e) at (\W, \H/2);
    \coordinate (#1-w) at (0, \H/2);
    \coordinate (#1-n) at (\W/2, \H);
    \coordinate (#1-s) at (\W/2, 0);
    \coordinate (#1-ne) at (\W, \H);
    \coordinate (#1-se) at (\W, 0);
  \end{scope}
}
\def\dx{3.2}
\def\dy{-2.0}

\begin{figure}[htbp]
\centering
\scalebox{0.75}{%
\begin{tikzpicture}[font=\small, >=Stealth,
  declare function={cs=1.6;},
]
\node[font=\scriptsize] at (0.56, 1.55) {rung $0$};
\node[font=\scriptsize] at (\dx+0.48, 1.55) {rung $1$};
\node[font=\scriptsize] at (2*\dx+0.40, 1.55) {rung $2$};
\node[font=\scriptsize] at (3.2*\dx, 1.55) {$\cdots$};
\node[font=\scriptsize] at (3.8*\dx+0.24, 1.55) {rung $\varepsilon n$};
\node[font=\tiny, text=blue!60!black] at (0.56, 1.25) {$n$ cols};
\node[font=\tiny, text=blue!60!black] at (\dx+0.48, 1.25) {$n{-}1$};
\node[font=\tiny, text=blue!60!black] at (2*\dx+0.40, 1.25) {$n{-}2$};
\node[font=\tiny, text=blue!60!black] at (4.0*\dx+0.24, 1.25) {$n{-}\varepsilon n$};
\drawmatrix[N00]{0}{0}{7}{10}{3}{3}{gray!15}
\drawmatrix[N01]{\dx}{0}{7}{9}{3}{2}{gray!15}
\drawmatrix[N02]{2*\dx}{0}{7}{8}{3}{1}{gray!15}
\node[font=\normalsize] at (3.6*\dx, 0.35) {$\cdots$};
\drawmatrix[N0k]{3.8*\dx}{0}{7}{7}{3}{0}{gray!15}
\draw[->, thick] (N00-e) -- (N01-w)
  node[midway,above,font=\tiny]{$x_1{=}0$};
\draw[->, thick] (N01-e) -- (N02-w)
  node[midway,above,font=\tiny]{$x_2{=}0$};
\draw[->] (N02-e) -- ++(0.5,0);
\node[font=\tiny, anchor=east] at (-0.15, 0.56) {$m$ rows};
\drawmatrix[N11]{\dx}{\dy}{6}{9}{2}{2}{orange!18}
\drawmatrix[N12]{2*\dx}{\dy}{6}{8}{2}{1}{orange!18}
\node[font=\normalsize] at (3.6*\dx, \dy+0.3) {$\cdots$};
\drawmatrix[N1k]{3.8*\dx}{\dy}{6}{7}{2}{0}{orange!18}
\draw[->, thick] (N11-e) -- (N12-w)
  node[midway,above,font=\tiny]{$x_2{=}0$};
\draw[->] (N12-e) -- ++(0.5,0);
\draw[->, orange!70!black, thick] (N00-s) -- (N11-n)
  node[midway,left,font=\tiny]{$x_1{=}1$};
\node[font=\tiny, anchor=east] at (-0.15+\dx, \dy+0.48) {$m{-}1$};
\drawmatrix[N22]{2*\dx}{2*\dy}{5}{8}{1}{1}{orange!30}
\node[font=\normalsize] at (3.6*\dx, 2*\dy+0.25) {$\cdots$};
\drawmatrix[N2k]{3.8*\dx}{2*\dy}{5}{7}{1}{0}{orange!30}
\draw[->] (N22-e) -- ++(0.5,0);
\draw[->, orange!70!black, thick] (N01-s) -- (N22-n)
  node[pos=0.3,right,font=\tiny]{$x_2{=}1$};
\draw[->, orange!70!black] (N11-s) -- (N22-n)
  node[pos=0.5,right,font=\tiny]{$x_2{=}1$};
\node[font=\tiny, anchor=east] at (-0.15+2*\dx, 2*\dy+0.40) {$m{-}2$};
\node[font=\Large] at (3.2*\dx, 2.7*\dy) {$\ddots$};
\drawmatrix[Nkk]{3.8*\dx}{3.8*\dy}{4}{7}{0}{0}{red!18}
\node[font=\tiny, anchor=east] at (-0.15+3.8*\dx, 3.8*\dy+0.32) {$m{-}\varepsilon n$};
\node[font=\tiny, text=blue!70!black]  at (0.24, 0.88) {$I$};
\node[font=\tiny, text=gray!70!black]  at (0.72, 0.88) {$D$};
\node[font=\tiny, text=red!60!black]   at (0.24, 0.24) {$C$};
\node[font=\tiny, text=green!50!black] at (0.72, 0.24) {$B$};
\node[font=\tiny, text=blue!70!black]  at (\dx+0.16, 0.88) {$I$};
\node[font=\tiny, text=green!50!black] at (\dx+0.72, 0.24) {$B$};
\node[font=\tiny, text=blue!70!black]  at (\dx+0.16, \dy+0.76) {$I$};
\node[font=\tiny, text=green!50!black] at (\dx+0.72, \dy+0.24) {$B$};
\draw[decorate, decoration={brace, amplitude=4pt}]
  (-0.1, 1.75) -- (4.0*\dx+1.2, 1.75)
  node[midway, above=5pt, font=\scriptsize, align=center]
  {$x_i{=}0$ spine: $\varepsilon n{+}1$ systems (same row count)};
\draw[decorate, decoration={brace, amplitude=5pt}]
  (3.8*\dx+1.6, 1.2) -- (3.8*\dx+1.6, 3.8*\dy-0.2)
  node[midway, right=6pt, font=\scriptsize, align=left]
  {$O((\varepsilon n)^2)$ positions\\[2pt]
   each: $\leq 2^r$ copies\\(same matrix, diff.\ RHS)};
\node[font=\tiny, text=blue!60!black, anchor=north, align=center]
  at (2*\dx+0.4, 2*\dy-0.2)
  {$x_1{=}0,\,x_2{=}1$ and $x_1{=}1,\,x_2{=}1$\\
   both reach $m{-}2$ rows (merge)};
\draw[->, gray!50, very thin] (N0k-e) -- ++(0.3,-0.1);
\draw[->, gray!50, very thin] (N0k-e) -- ++(0.3, 0.1);
\draw[->, gray!50, very thin] (N1k-e) -- ++(0.3,-0.1);
\draw[->, gray!50, very thin] (N1k-e) -- ++(0.3, 0.1);
\draw[->, gray!50, very thin] (N2k-e) -- ++(0.3,-0.1);
\draw[->, gray!50, very thin] (N2k-e) -- ++(0.3, 0.1);
\draw[->, gray!50, very thin] (Nkk-e) -- ++(0.3,-0.1);
\draw[->, gray!50, very thin] (Nkk-e) -- ++(0.3, 0.1);
\begin{scope}[shift={(0, 3.8*\dy - 1.0)}]
  \fill[blue!25] (0,0) rectangle (0.25,0.25);
  \draw[blue!50] (0,0) rectangle (0.25,0.25);
  \fill[blue!70] (0.125,0.125) circle (0.3mm);
  \node[font=\tiny, anchor=west] at (0.35, 0.125)
    {$I$: identity block (height \& width shrink per rung)};
  \fill[gray!8] (0,-0.15) rectangle (0.25,-0.40);
  \draw[thin,gray] (0,-0.15) rectangle (0.25,-0.40);
  \node[font=\tiny, anchor=west] at (0.35, -0.275)
    {$A$: top-right (height shrinks on $x_i{=}1$; width fixed)};
  \fill[red!12] (0,-0.55) rectangle (0.25,-0.80);
  \draw[thin,gray] (0,-0.55) rectangle (0.25,-0.80);
  \node[font=\tiny, anchor=west] at (0.35, -0.675)
    {$C$: rank-$r$ block (width shrinks; height $m{-}\varepsilon n$ fixed; controls RHS)};
  \fill[green!12] (0,-0.95) rectangle (0.25,-1.20);
  \draw[thin,gray] (0,-0.95) rectangle (0.25,-1.20);
  \node[font=\tiny, anchor=west] at (0.35, -1.075)
    {$B$: bottom-right (\emph{identical} in every system)};
  \fill[gray!15, opacity=0.20] (0.12,-1.35) rectangle ++(0.25,-0.25);
  \fill[gray!15, opacity=0.35] (0.06,-1.41) rectangle ++(0.25,-0.25);
  \fill[gray!15] (0,-1.47) rectangle ++(0.25,-0.25);
  \draw (0,-1.47) rectangle ++(0.25,-0.25);
  \node[font=\tiny, anchor=west] at (0.35, -1.57)
    {stacked sheets: $\leq 2^r$ copies (same matrix, different RHS)};
\end{scope}
\node[draw, rounded corners=3pt, fill=yellow!15, font=\small, align=center]
  at (0.9*\dx, 3.5*\dy)
  {total $\;\leq\;
   \underbrace{O((\varepsilon n)^{2})}_{\text{triangle positions}}
   \times \underbrace{2^{r}}_{\text{RHS variants}}
   = O((\varepsilon n)^{2}\cdot 2^{r})$};
\end{tikzpicture}}
\caption{Branching ladder as a triangular grid ($\varepsilon n$ rungs).
Blocks: \textcolor{blue!60}{$I$}~identity, $D$~top-right,
\textcolor{red!50}{$C$}~rank-$r$, \textcolor{green!40!black}{$B$}~fixed.
Horizontal arrows ($x_i{=}0$) keep all rows; orange arrows ($x_i{=}1$)
discard rows $1,\ldots,i$ and merge onto one of at most $2^r$ nodes.
Total: $O((\varepsilon n)^2\cdot 2^r)$.}
\label{fig:ladder}
\end{figure}

In the same spirit we also define a proof system $\LinTrees_{\F_q}$, which is p-equivalent to tree-like $\reslinneq{\F_q}$.

\begin{definition}\label{linTreesDef}
Let $\mathcal{F}=\{l_1(x_1,\ldots,x_n)=0, \ldots, l_m (x_1,\ldots,x_n)= 0\}$ be a linear system over $\F_q$ without solutions in the boolean cube $\{0,1\}^n\subset \F_q^n$.
A $\LinTrees_{\F_q}$ refutation $T$ of $\mathcal{F}$ is a tree such that:
\begin{itemize}
\item Every internal node $u$ is labeled by a linear form $f_u(x_1, \ldots, x_n)$ and outgoing edges of $u$ correspond to elements in $f_u(\{0,1\}^n)\subseteq \F_q$. 
If $a\in f_u(\{0,1\}^n)$ the corresponding edge is labeled by the equality $f_u = a$.
\item For a node $u$, let $\mathcal{G} = \{f_1=a_1, \ldots, f_k=a_k\}$ be a set of equations written on a path from the root to $u$. The node $u$ is a leaf iff the set of linear equations $\mathcal{F}\cup \mathcal{G}$ has no solutions over the whole $\F_q^n$.
\end{itemize}
\end{definition}

\begin{proposition}\label{peq}
$\LinDags_{\F_q}$ and $\reslinneq{\F_q}$ are p-equivalent as proof systems for
$\LinSys_{\F_q}$.
\end{proposition}

\begin{proposition}\label{peqT}
$\LinTrees_{\F_q}$ and tree-like $\reslinneq{\F_q}$ are p-equivalent as proof systems for
$\LinSys_{\F_q}$.
\end{proposition}

Proofs of Proposition~\ref{peq} and Proposition~\ref{peqT} are completely standard: a clause $C$ in a $\reslinneq{\F_q}$ (resp. tree-like $\reslinneq{\F_q}$) refutation
corresponds to a node marked with the system of equations $\neg C$ in $\LinDags_{\F_q}$ (resp. $\LinTrees_{\F_q}$) refutation. See, for example, \cite{PT18} for detailed
exposition of analogous correspondence between tree-like $\reslin{\F_q}$ and nondeterministic linear decision trees.

\subsection{Prover-Delayer games}

We now relate tree-like $\reslinneq{\F_q}$ to tree-like $\reslin{\F_q}$ by relating Prover-Delayer games in both cases. The games are defined as follows:

\begin{enumerate}
\item Tree-like $\reslin{\F_q}$ game on a set of 0-1 unsatisfiable linear equations $\mathcal{F}$.
\begin{itemize}
\item(Position) At every position there is a set of inequalities $$\mathcal{H}_{\neq} = \{h_1(x_1, \ldots, x_n) \neq c_1, \ldots, h_m(x_1, \ldots, x_n) \neq c_m\}$$
\item(Starting position) Game starts with $\mathcal{H}_{\neq} = \{0\neq a\}_{a\in \F_q, a \neq 0}$. 
\item(Round) Prover chooses inequalities $h \neq c \in \mathcal{H}_{\neq}$ and $f \neq a$, $g \neq b$ such that $(f - a) + (g - b) \equiv h - c$. Delayer either
chooses one of $f \neq a$, $g \neq b$ to be added to $\mathcal{H}_{\neq}$ or declares the position a branching point. In the latter case Prover chooses,
which of  $f \neq a$, $g \neq b$ will be added to $\mathcal{H}_{\neq}$.
\item(Endgame position) Game ends if $\mathcal{H}_{\neq}$ contains $x\neq 0$, $x\neq 1$ for some variable $x$, contains $f \neq a$ for some $f = a \in \mathcal{F}$, $a\in\F_q$ or 
contains $0 \neq 0$. 
\end{itemize}
\item $\LinTrees_{\F_q}$ game on a set of 0-1 unsatisfiable linear equations $\mathcal{F}$.
\begin{itemize}
\item(Position) At every position there is a set of equations $$\mathcal{H} = \{h_1(x_1, \ldots, x_n) = c_1, \ldots, h_m(x_1, \ldots, x_n) = c_m\}$$
\item(Starting position) Game starts with $\mathcal{H} = \mathcal{F}$. 
\item(Round) Prover chooses a linear form $f$. Delayer either chooses an equation $f = a$, $a\in f(\{0,1\}^n)$ to be added to $\mathcal{H}$ or declares
the position a branching point and chooses $a_1\neq a_2\in f(\{0,1\}^n)$. In the latter case Prover chooses,
which of  $f = a_1$, $f = a_2$ will be added to $\mathcal{H}$.
\item(Endgame position) Game ends if $\mathcal{H}$ has no solutions in $\F_q^n$.
\end{itemize}
\end{enumerate}

It is easy to see that if there exists Delayer's strategy guaranteeing selection of $h$ branching points on $\mathcal{F}$ then for every
$\LinTrees_{\F_q}$ (resp.\ tree-like $\reslin{\F_q}$) refutation $T$ of $\mathcal{F}$ there exists an embedding
of the full binary tree of depth $h$ into $T$.

\begin{proposition}
If there exists a strategy with a starting position $\mc F$ for Delayer in the tree-like
$\reslin{\F_q}$ game (respectively, $\LinTrees_{\F_q}$ game) that guarantees at least $h$ branching points, then the size of
a tree-like $\reslin{\F_q}$ (respectively $\LinTrees_{\F_q}$) refutation of $\mc F$
must be at least $2^h$.
\end{proposition}

The proof is completely analogous to the proof of Lemma 31 in \cite{PT18}. See \cite{PT18} for more details on Prover-Delayer games for tree-like $\reslin{\F_q}$.

\begin{theorem}\label{delayerEquiv}
Let $\mathcal{F}$ be 0-1 unsatisfiable linear system. If there exists Delayer's strategy guaranteeing selection of $h$ branching
points for $\LinTrees_{\F_q}$ Prover-Delayer game on $\mathcal{F}$ then there exists a Delayer's strategy guaranteeing selection of $h$ branching points for tree-like $\reslin{\F_q}$ game as well.
\end{theorem}
\begin{proof}
We play two games simultaneously. When Prover makes a decision in tree-like 
$\reslin{\F_q}$ game we make a decision for Prover in $\LinTrees_{\F_q}$ game, see what decides
Delayer and make a decision for Delayer in tree-like $\reslin{\F_q}$ game.

For every inequality $f\neq a$ added in tree-like $\reslin{\F_q}$ game there will be exactly one equation $f = b$, where $a\neq b$, added in $\LinTrees_{\F_q}$ game. It is easy
to see that tree-like $\reslin{\F_q}$ game cannot end earlier than $\LinTrees_{\F_q}$ game.
After $\LinTrees_{\F_q}$ game ends Delayer continues arbitrarily choosing equations. The number of branching points will coincide in two games.

Assume Prover in tree-like $\reslin{\F_q}$ game chooses $h \neq a$ among added inequalities and $f \neq b, g \neq c$ such that $(f - b) + (g - c) \equiv h - a$. By our induction hypothesis there must be equation $h = a'$ added in $\LinTrees_{\F_q}$ game. Prover in 
$\LinTrees_{\F_q}$ game chooses linear form $f$.

If Delayer in $\LinTrees_{\F_q}$ game chooses equality $f = b'$, then Delayer in tree-like $\reslin{\F_q}$ game chooses $f \neq b$ iff $b\neq b'$ otherwise it chooses $g\neq c$.
Note that after addition of $f = b'$ in $\LinTrees_{\F_q}$ game equation $g = c + a' - a + b - b'$ must be in the span of added equations and $\mathcal{F}$.

If Delayer in $\LinTrees_{\F_q}$ game declares the current position a branching point and chooses $f = b_1$, $f = b_2$ equations, then Delayer
in tree-like $\reslin{\F_q}$ game declares the position a branching point and if Prover chooses $f \neq b$ or $g\neq c$ then make a corresponding choice for Prover in $\LinTrees_{\F_q}$ game
(that is so that $f = b'$ for $b \neq b'$ is added if $f\neq b$ was chosen or $g = c'$ for $c \neq c'$ is added if $g\neq c$ was chosen).

\end{proof}

\section{Hard instances based on error correcting codes}\label{hardInstSec}

Denote $\ECC^{n, k, d}_{\F_q} \subset \LinSys_{\F_q}$ the set of instances 
$(A, b)\in \LinSys_{\F_q}$ such that $A$ is generator $k \times n$ matrix for ECC
with parameters 
$(n, k, d)$. For a lower bound for $\BinRegDags_{\F_q}$ refutations in Section~\ref{bindagLB} we will need a stronger notion of a \textbf{robust} ECC instance. See discussion after Definition~\ref{binDagsDef} for motivating example.

\begin{definition}[Robust linear systems]\label{def:robust}
Let an instance $(A, b) \in \ECC^{n, k, d}_{\F_q}$ be written as $\mc F =\{f_1 = a_1, \ldots, f_m = a_m\}$. Consider pairs $(\rho, \mc G)$
such that $\rho$ is a partial assignment with $|supp(\rho)| = s$, where $supp(\rho)$
denotes the set of variables on which $\rho$ is defined,
and $\mc G$ is a linear system such that:
\begin{enumerate} 
\item\label{sub} $\mc G \subseteq \langle \mc F\rangle$.
\item $\mc G\rst_{\rho}$ is 0-1 unsatisfiable.
\item $\mc G_{[supp(\rho)]}$ depends on all variables in $supp(\rho)$.
\footnote{$\mc G_{[I]}$ is the submatrix formed by
the columns corresponding to variables in $I$. See Section~\ref{sec:notation} for the general definition.}
\end{enumerate}
Assume that for all pairs $(\rho, \mc G)$ satisfying the conditions above,
$$dim\!\left(\langle \mc G_{[supp(\rho)]}\rangle\right)\geq r.$$
Here $\mc G_{[supp(\rho)]}$ is the matrix obtained from (the coefficient matrix of) $\mc G$
by retaining only the columns indexed by $supp(\rho)$, viewed as a set of row vectors;
its dimension $dim(\langle \mc G_{[supp(\rho)]}\rangle)$ is the rank of this submatrix.
Then $(A, b)$ is called \textbf{$(s, r)$-robust}.
\end{definition}

\begin{figure}[htbp]
\centering
\begin{tikzpicture}[font=\small, >=Stealth,
  cell/.style={draw, minimum size=6.5mm, inner sep=0pt},
  hcell/.style={draw, minimum size=6.5mm, inner sep=0pt, fill=orange!30},
  gcell/.style={draw, minimum size=6.5mm, inner sep=0pt, fill=gray!12},
  bcell/.style={draw, minimum size=6.5mm, inner sep=0pt, fill=blue!10}
]

\foreach \i/\lbl in {0/$x_1$, 1/$x_2$, 2/{$\cdots$}, 3/$x_{i_1}$,
                     4/{$\cdots$}, 5/$x_{i_s}$, 6/{$\cdots$}, 7/$x_n$}{
  \node[font=\scriptsize] at (\i*0.75 - 2.6, 2.1) {\lbl};
}
\node[font=\scriptsize] at (8*0.75 - 2.6 + 0.25, 2.1) {$b$};

\foreach \row in {0,1,2,3}{
  \foreach \col in {0,1,6,7}{
    \node[gcell] at (\col*0.75 - 2.6, 1.5 - \row*0.75) {};
  }
  \node[hcell] at (3*0.75 - 2.6, 1.5 - \row*0.75) {};
  \node[hcell] at (5*0.75 - 2.6, 1.5 - \row*0.75) {};
  \node[font=\scriptsize] at (2*0.75 - 2.6, 1.5 - \row*0.75) {$\cdots$};
  \node[font=\scriptsize] at (4*0.75 - 2.6, 1.5 - \row*0.75) {$\cdots$};
  \node[bcell] at (8*0.75 - 2.6 + 0.25, 1.5 - \row*0.75) {};
}
\node[font=\scriptsize] at (-2.6, 1.5 - 3.75*0.75) {$\vdots$};
\node[font=\scriptsize] at (8*0.75 - 2.6 + 0.25, 1.5 - 3.75*0.75) {$\vdots$};
\node[font=\scriptsize, left=2pt] at (-2.6 - 0.375, 1.5 - 1.5*0.75)
  {$\mathcal{G}\subseteq\langle\mathcal{F}\rangle$\;\;};

\draw[thick] (7*0.75 - 2.6 + 0.375 + 0.08, 1.5 + 0.375)
          -- (7*0.75 - 2.6 + 0.375 + 0.08, 1.5 - 3*0.75 - 0.375);

\draw[orange!70!black, decorate,
  decoration={brace, amplitude=5pt, mirror}]
  (3*0.75 - 2.6 - 0.35, 1.5 - 3*0.75 - 0.45)
  -- (5*0.75 - 2.6 + 0.35, 1.5 - 3*0.75 - 0.45)
  node[midway, below=5pt, align=center, orange!80!black]
    {$\mathcal{G}_{[\mathrm{supp}(\rho)]}$\quad
     $\mathrm{rank}\!\left(\mathcal{G}_{[\mathrm{supp}(\rho)]}\right)\!\geq r$};

\draw[orange!70!black, decorate,
  decoration={brace, amplitude=5pt}]
  (3*0.75 - 2.6 - 0.35, 2.4)
  -- (5*0.75 - 2.6 + 0.35, 2.4)
  node[midway, above=5pt, orange!80!black, font=\scriptsize]
    {$\mathrm{supp}(\rho)$, $|\mathrm{supp}(\rho)|=s$};

\draw[->, thick] (1.5, 0) -- (6.3, 0)
  node[midway, above right, font=\scriptsize] {\, \, apply $\rho$};

\node[draw, rounded corners=4pt, fill=red!8, align=center,
  minimum width=2.4cm, minimum height=1.4cm] at (7.5, 0)
  {$\mathcal{G}\rst_\rho$\\[3pt]0-1 unsat.};

\end{tikzpicture}
\caption{Illustration of $(s,r)$-robustness.
An instance $(A,b)$ (denoted as space $\mc F$) is $(s,r)$-robust if for every partial assignment $\rho$
with $|\mathrm{supp}(\rho)|=s$ (orange columns) and every subsystem
$\mathcal{G}\subseteq\langle\mathcal{F}\rangle$ whose restriction
$\mathcal{G}\rst_\rho$ is 0-1 unsatisfiable and depends on all variables
in $\mathrm{supp}(\rho)$, the rank of the submatrix
$\mathcal{G}_{[\mathrm{supp}(\rho)]}$ is at least $r$.
Intuitively: any $s$-variable witness of 0-1 infeasibility must involve at
least $r$ independent equations on those variables.}
\label{fig:robustness}
\end{figure}

\subsection{Emptyness of $\ECC^{n, k, d}_{\F_q}$}

Sometimes $\ECC^{n, k, d}_{\F_q}$ is empty even if there exist ECCs with parameters 
$(n, k, d)$. Characterization of when $\ECC^{n, k, d}_{\F_q}$ is empty is crucial for
the construction of robust instances in Section~\ref{sec:robust} and the tree-like lower
bound in Section~\ref{sec:tree-like}. 

Results of this section are summarized in the following theorem:

\begin{theorem}\label{thm:isemptyECC}
Let $q$ be a prime power. The following holds:
\begin{enumerate}
\item\label{itm:nonempECC}\textbf{Nonemptyness.} If there exists an ECC over $\F_q$
with parameters $(n, k, d)$, generator matrix $A$ and $n < (\log_2 q) k$ then there
exists $b \in \F_q^k$ such that $(A, b) \in \ECC^{n, k, d}_{\F_q}$.
\item\label{itm:empECC}\textbf{Emptyness.} If $d \geq (q \ln q) k^3$ then $\ECC^{n, k, d}_{\F_q}$ is empty.
\end{enumerate}
\end{theorem}

Theorem~\ref{thm:isemptyECC}.(\ref{itm:nonempECC}) is trivial. Let $A$ be a $k \times n$ generator matrix for the ECC from the statement. The image $A(\{0,1\}^n)$ of 0-1 points
under $A$ has size at most
$2^n < q^k = |\F_q^k|$. Therefore there exists $b \in \F_q^k$, $b \notin A(\{0,1\}^n)$
and thus $(A, b) \in \ECC^{n, k, d}_{\F_q}$.

Unfortunately, this simple counting argument is the only way to
construct 0-1 unsatisfiable instances we have up to now in case $d$ is not too
small ($d$ superlinear in $\mathrm{char}(\F_q)$).  Although we can pick explicit
$A$ for our lower bounds, we can only prove \emph{existence} of suitable $b$
without explicitly specifying it.  A natural concrete question arises for the
Reed-Solomon code, which has a particularly simple generator matrix (though $q$
is too large in that case to be directly relevant for our lower bounds):

\begin{problem}[\textbf{Open}]\label{problem:free}
Choose an arbitrary ordering on nonzero elements
$a_1, \ldots, a_{q - 1} \in \F_{q}^*$.  The Reed-Solomon code over $\F_{q}$ is
the linear code
$\mc C_{RS} := \{(p(a_1), \ldots, p(a_{q - 1}))\mid p\in \F_{q}[x],\ \deg(p)\leq k\}$.
Since a polynomial of degree at most $k$ has at most $k$ roots, the minimal
distance of $\mc C_{RS}$ is $d = q - k$.  Pick a concrete $k \times n$ generator
matrix $A_{RS}$ for $\mc C_{RS}$, for instance
$A_{RS}^{i,j} := \alpha^{(i-1)(j-1)}$ where $\alpha \in \F_q$ is the generator
of the multiplicative group $\F_q^*$.  What is a concrete $b\in \F_q^k$ such that
$b \notin A_{RS}(\{0,1\}^n)$, if it exists?
\end{problem}

The second part of the theorem partially characterizes 0-1 solvability of linear systems 
$A \cdot x = b$ based on ECCs: if the minimal distance $d_A$ is big enough compared
to the number of equations $k$ (that is if $d_A \geq (q \ln q) k^3$), then 
the system is necessarily 0-1 solvable. The rest of this section is devoted to Theorem~\ref{thm:isemptyECC}.(\ref{itm:empECC}) and its proof.

Consider the $k=1$ case where there is just one equation
$a_1 x_1 + \cdots + a_n x_n = b$ over $\F_q$.  Such an equation is always 0-1
satisfiable whenever the number of nonzero coefficients $d > q$.  It is natural
to ask whether something analogous holds for $k > 1$ linear equations:

\begin{questn}\label{quest:delta}
Denote $\Delta(k, q) \in \mathbb{N}\cup \{+\infty\}$ the minimal number such
that for every $k\times n$ matrix $A$ over $\F_q$ whose row code has minimal
distance $d_A \geq \Delta(k, q)$, the system $A\cdot x = b$ is 0-1 satisfiable
for all $b\in \F_q^k$.  How does $\Delta(k, q)$ grow?  Does
$\Delta(k, q) < +\infty$ hold?
\end{questn}

Part~(\ref{itm:empECC}) of Theorem~\ref{thm:isemptyECC} gives the upper bound
$\Delta(k, q) \leq ((q+1) \ln q)\, k^3$.  Its proof reduces
Question~\ref{quest:delta} to the following additive-combinatorics lemma.

The key to the proof is the following lemma almost completely resolving a natural
question in additive combinatorics: how large is $t_0(k, q)$ such that for any
family $\mc X = \{X_i\}$ of bases in $\F_q^k$ if $|\mc X| \geq t_0(k, q)$ then
$X_1+_M\ldots +_MX_{|\mc X|} = \F_q^k$ where $+_M$ is the Minkowski sum.

\begin{lemma}\label{addComb}
Let $X_1, \ldots, X_t \subset \F_q^k$ be $t$ bases of $\F_q^k$. If $t \geq ((q+1) \ln q) k^2$ then $X_1+_M\ldots+_MX_t = \F_q^k$.
\end{lemma}
We first prove Theorem~\ref{thm:isemptyECC}.(\ref{itm:empECC}) using Lemma~\ref{addComb}
and after that proceed to the proof of Lemma~\ref{addComb}.

\begin{proof}(of Theorem~\ref{thm:isemptyECC}.(\ref{itm:empECC}))
Consider a system $A \cdot x = b$ such that $d_A \geq (q \ln q) k^3$. Since $d_A > 0$,
$A$ must have full row rank (rank $= k$). Let $X_1$ be an $k$-element set of linearly independent columns in $A$. Denote $A_1$ the matrix obtained by removing
$X_1$ from $A$. For $A_1$ holds $d_{A_1} \geq (q \ln q) k^3 - k$.

We can successively apply the procedure above at least $(q \ln q)k^2$ times and obtain a sequence $X_1, \ldots, X_t$, $t \geq (q \ln q)k^2$ of disjoint $k$-element subsets of 
linearly independent columns of $A$. By Lemma~\ref{addComb} $X_1+_M\ldots+_MX_t = \F_q^k$, therefore $A(\{0, 1\}^n)=\F_q^k$ and thus $A\cdot x = b$ must be 0-1 satisfiable.
\end{proof}

\begin{proof}(of Lemma~\ref{addComb})
We first prove the following claim:
\begin{claim}\label{blockClaim}
Let $S$ be a set of vectors in $\F_q^k$ such that $|S| < k$. Let $X_1, \ldots, X_t \subset \F_q^k$ be $t$ bases of $\F_q^k$. If $t \geq (q + 1) (\ln q)k$ then
there exists $v\in \F_q^k$ linearly independent from vectors in $S$ and some $a\in \F_q^k$ such that $\{a + \alpha v\, |\, \alpha \in \F_q\} \subseteq X_1+_M\ldots+_MX_t$.
\end{claim}
\begin{proof}
Let $Y_i = X_1 +_M \dots +_M X_i$, $i < t$. Pick some $v \in X_{i + 1}$ such that $v\notin \langle S\rangle$ and assume there does not exist $a\in \F_q^k$ 
such that $\{a + \alpha v\, |\, \alpha \in \F_q\} \subseteq Y_i$. Then consider the border $\partial_v Y_i := \{b\, |\, b\in Y_i, b + v\notin Y_i\}$ of $Y_i$ 
with respect to $v$. Since no complete line parallel to $v$ is in $Y_i$ for every $b\in Y_i$ can be mapped to $\partial_v Y_i$ by a shift $b + \alpha_b v \in \partial_v Y_i$ and there 
are at most $q - 1$ of elements in $Y_i$ that are mapped to the same vector in $\partial_v Y_i$. Therefore $|\partial_v Y_i| \geq \frac{1}{q}\cdot |Y_i|$. Now note that 
$|Y_i +_M X_{i+1}| \geq |Y_i| + |\partial_v Y_i|\geq (1 + \frac{1}{q})\cdot |Y_i|$ where
the first inequality holds because $\partial_v Y_i\subseteq Y_i$ shifted by
$v \in X_{i+1}$ does not intersect $Y_i$. 

If our assumption on nonexistence of certain lines in $Y_i$ were true for all steps up to $i$, then $|Y_i|\geq (1 + \frac{1}{q})^i$. Therefore $(1 + \frac{1}{q})^i < q^k$ and
thus $i < (q + 1)(\ln q)k$. Since $t \geq (q + 1)(\ln q)k$ a line like in the statement of the claim must exist in $Y_t$.
\end{proof}

We split the sequence $X_1, \ldots, X_t$ into $k$ blocks $X^{(i)}_1, \ldots, X^{(i)}_{t_i}$ of size $\geq (q + 1)(\ln q)k$. We define $S_i := \{v_1, \ldots, v_i\}$ inductively.
Let $S_0$ be empty. For $i \geq 1$ by the Claim~\ref{blockClaim} there exists $v_i$ linearly independent from vectors in $S_{i - 1}$ and some $a_i\in \F_q^k$ such that 
$\{a_i + \alpha v_i\, |\, \alpha \in \F_q\} \subseteq X^{(i)}_1+_M\ldots+_MX^{(i)}_{t_i}$. Since vectors in $S_m$ are linearly independent, the sum of these lines 
gives $\F_q^k$. On the other hand, the sum of lines is in $X_1+_M\ldots+_MX_t$.

\end{proof}

From Theorem~\ref{thm:isemptyECC}.(\ref{itm:nonempECC}) it follows that if
$q > 2$ then $\Delta(k, q) = \Omega(k \log q)$ (by the Gilbert-Varshamov bound,
for instance).  We currently do not have better lower bounds on $\Delta(k, q)$.

\begin{problem}[\textbf{Open}]
Narrow down the interval $[\Omega(k \log q),\, ((q+1) \ln q)\, k^3]$ for
$\Delta(k, q)$.
\end{problem}

\subsection{Construction of robust instances}\label{sec:robust} 

In this section we use results of the previous section to prove existence of robust instances.

The construction of $(s,r)$-robust instances in Theorem~\ref{thm:robust} below
relies on the distance condition $d > 2n/3$ together with Theorem~\ref{thm:isemptyECC}.(\ref{itm:empECC}).  The idea is as follows.  Set $s := n/3$.  Since $n - s < d$, every $s$ columns of $A$ span a full-rank submatrix (rank equals the number of columns); for any partial assignment $\rho$ with $|\mathrm{supp}(\rho)|=s$ and any subsystem $\mc G \subseteq \langle\mc F\rangle$ the submatrix $\mc G_{[\mathrm{supp}(\rho)]}$ inherits full rank, and Theorem~\ref{thm:isemptyECC}.(\ref{itm:empECC}) then forces $\dim(\langle\mc G_{[\mathrm{supp}(\rho)]}\rangle)\geq r$.  The nature of the problem changes for smaller $s$, say $s \leq k/2$: the submatrix $A_{[\mathrm{supp}(\rho)]}$ no longer has full rank, $\mc G_{[\mathrm{supp}(\rho)]}$ can have arbitrarily small rank, and the choice of $b$ becomes important (cf.\ Problem~\ref{problem:free}).

\begin{problem}[\textbf{Open}]
Do $(s, \omega(\log n))$-robust instances exist in $\ECC^{n,k,d}_{\F_q}$ for
$s \leq k / 2$?
\end{problem}

\begin{theorem}\label{thm:robust}
If a 0-1 unsatisfiable instance $A\cdot x = b$ is such that $A$ is a $k\times n$ 
generator matrix for 
a code over $\F_q$ with distance $d$ satisfying $d > 2n/3$ then it is 
$(n / 3, \Omega\left((n/(q + 1)\ln q)^{1/3}\right))$-robust.
\end{theorem}
\begin{proof}
Let $s := n / 3$. Since $n - s < d$, any $s$ columns of $A$ form a $k\times s$
submatrix of full column rank (rank $= s$; since otherwise there exists
$x\in \mc C\setminus 0$ with $\omega(x) < d$). Let $\rho$ be a partial assignment with the support $I := supp(\rho)$ such
that $|I| = s$. Since the $k\times s$ submatrix $A_I$ of $A$ formed by the columns
indexed by $I$ has full column rank (rank $= s$), for any $\mc G \subseteq \langle \mc F\rangle$
holds 
$dim(\langle \mc G_{[I]}\rangle) = dim(\langle \mc G\rangle) \geq dim(\langle (\mc G\rst_{\rho})_{[1]}\rangle)$.
By Theorem~\ref{thm:isemptyECC}.(\ref{itm:empECC}), since 
$\omega(\langle (\mc G\rst_{\rho})_{[1]}\rangle) \geq d - s \geq n / 3$,
if $\mc G\rst_{\rho}$ is 0-1 unsatisfiable then 
$dim(\langle (\mc G\rst_{\rho})_{[1]}\rangle) \geq (n / (3 q \ln q))^{1/3} =: r$ and thus
$dim(\langle (\mc G_{[I]})\rangle) \geq r$. It follows that $(A, b)$ is 
$(s, r)$-robust.
\end{proof}

Theorem~\ref{thm:robust} together with Corollary~\ref{cor:rand_code_dist} and
Proposition~\ref{prop:hermit_code_dist} imply the following.

\begin{corollary}\label{cor:rand_robust}
If $A$ is uniformly random $k \times n$ matrix over $\F_q$ where 
$n / \log q < k < n / 9$ 
then there exists $b\in \F_q^k$ such that $A\cdot x = b$ is 0-1 unsatisfiable and
with high probability $A\cdot x = b$ is 
$(n / 3, \Omega\left((n/(q + 1)\ln q)^{1/3}\right))$-robust.
\end{corollary}

\begin{corollary}\label{cor:hermit_robust}
Let $A$ be a generator matrix of the Hermitian code over $\F_q$
for $q = p^2$ for a prime $p$, with parameters $n = p^3$, $k = \lfloor p^3/\log_2 p - p^2 / 2 \rfloor$,
$d \geq (1 - 1/\log p) n$. Then there exists $b\in \F_q^k$ such that $A\cdot x = b$
is 0-1 unsatisfiable $(n / 3, \Omega\left((n/(q + 1)\ln q)^{1/3}\right))$-robust
instance.
\end{corollary}

\section{Lower bounds for $\reslin{\F_q}$ refutations}

\subsection{Linear splitting tree refutations}\label{sec:tree-like}

We begin with a high-level overview of the Delayer's strategy and sketch the argument why it guarantees
a good number of branching points. We postpone the details to the Section~\ref{sec:TLproofs}.

\begin{figure}[htbp]
\centering
\begin{tikzpicture}[
  rnd/.style={circle, draw, minimum size=7mm, inner sep=1pt, font=\small},
  brn/.style={rectangle, rounded corners=3pt, draw, thick, fill=orange!20,
              minimum width=8mm, minimum height=8mm, inner sep=2pt, font=\small},
  lf/.style={draw=none, font=\small},
  hl/.style={red!70!black, very thick},  
  >=Stealth
]
\node[rnd]        (r)   at ( 0,   0  ) {};      
\node[brn]        (b1)  at ( 0,  -1.6) {$\square$}; 
\node[brn]        (b2l) at (-2.2,-3.2) {$\square$}; 
\node[rnd]        (nb)  at ( 2.2,-3.2) {};      
\node[brn]        (b2r) at ( 2.2,-4.8) {$\square$}; 
\node[lf]         (l1)  at (-3.2,-4.8) {$\bot$};
\node[lf]         (l2)  at (-1.2,-4.8) {$\bot$};
\node[lf]         (l3)  at ( 1.2,-6.1) {$\bot$};
\node[lf]         (l4)  at ( 3.2,-6.1) {$\bot$};

\draw[hl] (r) -- (b1);
\draw[hl] (b1) -- node[left,font=\scriptsize,red!60!black]{$f{=}a_1$} (b2l);
\draw[hl] (b2l) -- (l1);

\draw (b1) -- node[right,font=\scriptsize]{$f{=}a_2$} (nb);
\draw (b2l) -- (l2);
\draw (nb)  -- (b2r);
\draw (b2r) -- (l3);
\draw (b2r) -- (l4);

\node[draw=none, font=\scriptsize, orange!70!black, right=3pt] at (b1.east)
  {\#1};
\node[draw=none, font=\scriptsize, orange!70!black, left=3pt] at (b2l.west)
  {\#2};
\node[draw=none, font=\scriptsize, orange!70!black, right=3pt] at (b2r.east)
  {\#2};


\node[rnd, label=right:{\scriptsize non-branching}] at (-4.5,-5.4) {};
\node[brn, label=right:{\scriptsize branching ($\square$)}] at (-4.5,-6.2) {$\square$};

\node[draw=none, font=\scriptsize, align=center] at (0.6,-6.9)
  {Every path has $\geq 2$ branching points
   $\;\Rightarrow\;$ size $\geq 2^2 = 4$};
\end{tikzpicture}
\caption{A $\LinTrees_{\F_q}$ refutation subtree.  Circles are non-branching positions
where Delayer forces the choice of the next node;
shaded squares~($\square$, labelled \#1, \#2) are \emph{branching points}
declared by Delayer where Prover can choose one of two nodes as the next node; $\bot$~marks contradiction leaves.  The size lower
bound $2^h$ follows from the \emph{minimum} number of branching points along
any root-to-leaf path (highlighted in red: 2 branching points on this path).
Every such path sees at least $h$ branching points, forcing a complete binary
tree of depth~$h$ to embed.}
\label{fig:game-tree}
\end{figure}

Assume $A\cdot x = b$ is the starting system.
For brevity in this overview, at each node $u$ of the game tree we write:
\begin{itemize}[nosep]
\item $E_u$ for all equations added during the game so far,
\item $S_u$ for the system of equations of weight $\leq \tau:=\Theta(d_A^{4/5})$ in
      $\mathrm{span}(A\cdot x{=}b,\,E_u)$ (the ``short'' subsystem at $u$),
\item $E^\bullet_u\subseteq E_u$ for the branching equations only, and
\item $S^\bullet_u$ for the system of weight-$\leq\tau$ equations in
      $\mathrm{span}(A\cdot x{=}b,\,E^\bullet_u)$.
\end{itemize}

\noindent\textbf{Core idea.}
Keep $S_u$ 0-1 satisfiable for as long as possible.
Once $S_u$ is 0-1 unsatisfiable and depends on few variables, a short
refutation of $S_u$---and hence of $A\cdot x=b$---exists.

\noindent\textbf{Strategy (first attempt).}
When Prover picks a form $f$, let $f'$ be its minimal-weight representative
modulo $\mathrm{span}(A\cdot x{=}b, E_u)$ (so $f'\equiv\alpha f+h$ for some
$h\in\mathrm{span}(A\cdot x{=}b,E_u)$, $\alpha\neq 0$, and adding $f'=c$ is equivalent
to adding $f=\tilde{c}$ for a corresponding $\tilde{c}$):
\begin{enumerate}
\item\label{nonbranchX}\textbf{Non-branching.}
  If $S_u\models_{0,1}f'=c$ for some $c\in\F_q$, add $f=\tilde{c}$ and continue.
\item\label{branchX}\textbf{Branching.}
  Otherwise pick $c_1\neq c_2$ such that $S_u\wedge f'=c_i$ is 0-1 satisfiable
  for $i=1,2$; declare a \emph{branching point} and let Prover choose
  $f=\tilde{c}_i$. Call this a \emph{branching equation}.
\end{enumerate}

\begin{figure}[htbp]
\centering
\scalebox{0.80}{%
\begin{tikzpicture}[
  box/.style={draw, rectangle, rounded corners=3pt, minimum width=5.8cm,
              minimum height=9mm, align=center, font=\small},
  dec/.style={draw, diamond, aspect=3.2, minimum width=5.8cm,
              minimum height=9mm, align=center, font=\small},
  >=Stealth, node distance=1.1cm
]
\node[box] (start)
  {Prover picks $f$;\quad compute min-weight $f'\equiv f \bmod \mathrm{span}(A\cdot x{=}b,E_u)$};
\node[dec, below=of start] (q)
  {$S_u\;\models_{0,1}\;f'{=}c\;$?};
\node[box, below left=1.3cm and 2.2cm of q, fill=green!8] (nb)
  {\textbf{Non-branching.}\\[2pt]
   Add $f = \tilde{c}$; no branching point.};
\node[box, below right=1.3cm and 1.2cm of q, fill=orange!15] (br)
  {\textbf{Branching point.}\\[2pt]
   Pick $c_1\!\neq\!c_2$ with $S_u\wedge f'{=}c_i$ sat.\\
   Prover adds $f = \tilde{c}_i$.};
\draw[->] (start) -- (q);
\draw[->] (q) -- node[above left, font=\scriptsize]{yes} (nb);
\draw[->] (q) -- node[above right,font=\scriptsize]{no}  (br);
\end{tikzpicture}}
\caption{Delayer's strategy at each round.  The min-weight representative
$f'$ of $f$ modulo ${\mathrm{span}(A\cdot x{=}b,E_u)}$ is used to decide the case.
A branching point occurs exactly when neither value of $f'$ is 0-1 forced
by the short system~$S_u$.}
\label{fig:strategy}
\end{figure}

\noindent\textbf{The problem.}
Adding $f=\tilde{c}$ can pull several new short equations into $S_u$
(Figure~\ref{fig:span-decomp}), potentially making $S_{u'}$
0-1 unsatisfiable at the child node $u'$.
Let $H$ denote the new short equations entering $S_{u'}$.
Although $\mathrm{span}(S_u,H)$ can have small minimal distance,
we rescue 0-1 satisfiability via the following key claim.
Intuitively, it finds a partial sub-assignment $\rho\subset\rho_0$ that
zeroes out all new short equations, leaving a restricted system with
large minimal distance and no dimension gain. See Figure~\ref{fig:span-decomp}

\begin{figure}[htbp]
\centering
\scalebox{0.80}{%
\begin{tikzpicture}[font=\small, >=Stealth]

\draw[fill=blue!5, draw=blue!40]
  (0,0) ellipse (5.4cm and 2.4cm);
\node[blue!50!black] at (0, 1.6)
  {$\mathrm{span}(A\!\cdot\!x{=}b,\;E_u)$\quad\textit{(all reachable equations)}};

\draw[fill=blue!16, draw=blue!55]
  (-1.4, -0.15) ellipse (2.0cm and 1.4cm);
\node[blue!65!black, align=center] at (-1.8, 0.15)
  {$\mathrm{span}(A\!\cdot\!x{=}b)$};

\draw[fill=orange!28, draw=orange!75]
  (2.5, -0.3) ellipse (1.8cm and 0.95cm);
\node[align=center] at (2.2, -0.3)
  {$S_u$\\[-2pt]{\scriptsize weight~$\leq\tau$, 0-1 sat.}};
  
\draw[fill=red!28, draw=orange!75]
  (4.0, -0.3) ellipse (0.1cm and 0.1cm);
\draw[fill=red!28, draw=orange!75]
  (3.8, -0.5) ellipse (0.1cm and 0.1cm);  

\node[circle, draw=red!70, fill=red!25, minimum size=6mm, inner sep=1pt]
  (neq) at (5.8, 2.1) {};
\node[above right=1pt of neq, font=\scriptsize] {$f{=}\tilde{c}$\;\textit{(new)}};

\draw[->, red!70!black, thick, dashed]
  (neq) to[bend left=20]
  node[right, font=\scriptsize, red!60!black, align=left]
    {\, can pull several\\[-2pt]\, new short\\[-2pt]\, equations}
  (4.2, -0.25);

\end{tikzpicture}}
\caption{Adding a new equation $f{=}\tilde{c}$ (red) expands the span and
can import new short equations into $S_u$, potentially destroying its
0-1 satisfiability.  Note that $S_u$ lives in
$\mathrm{span}(A\cdot x{=}b,E_u)$ but outside of
$\mathrm{span}(A\!\cdot\!x{=}b)$: short equations can only arise from
combining game equations $E_u$ with the original system.
Claim~\ref{clm:narrowSat} finds $\rho\subset\rho_0$ that zeroes out
all $\leq\tau_0$-short equations (condition~1), ensures large minimal distance
(condition~2), and bounds the dimension (condition~3), so that
Theorem~\ref{thm:isemptyECC} applies.}
\label{fig:span-decomp}
\end{figure}

\begin{claim}\label{clm:narrowSat}
Let 0-1 assignment $\rho_0$ satisfy $S_u\wedge f'=c$, and set $\tau_0:=\tau^\delta$ for
some fixed $\delta<1$.
There exists $\rho\subset\rho_0$ such that:
\begin{enumerate}
\item All equations of weight $\leq\tau_0$ in $\mathrm{span}(S_u,H)$
  are satisfied by $\rho$.
  \quad{\normalfont\itshape(Short equations are zeroed out.)}
\item Remaining equations of $(S_u\wedge H)\rst_\rho$ have weight $\geq\tau_0$.
  \quad{\normalfont\itshape(Restricted system has large minimal distance.)}
\item $\dim\bigl((S_u\wedge H)\rst_\rho\bigr)\leq\dim(E_{u'})$.
  \quad{\normalfont\itshape(No dimension gain.)}
\end{enumerate}
\end{claim}

Conditions~(2)--(3) let us apply Theorem~\ref{thm:isemptyECC}.(\ref{itm:empECC}):
if $\dim(E_{u'})\leq(\tau_0/(q{+}1)\ln q)^{1/3}$ then
$(S_u\wedge H)\rst_\rho$ is 0-1 satisfiable, hence $S_{u'}$ is 0-1 satisfiable.

We can therefore conclude that at every endgame position
$\dim(E_u)\geq(\tau_0/(q{+}1)\ln q)^{1/3}$.
However, $\dim(E_u)$ is not bounded by the number of branching points,
since non-branching equations also raise the dimension.

\medskip
\noindent\textbf{Modified strategy.}
Replace $S_u$ with $S^\bullet_u$ throughout: track only the short equations
derivable from $A\cdot x{=}b$ and the branching equations $E^\bullet_u$.
Then $\dim(E^\bullet_u)$ equals the number of branching points so far,
which is exactly what we want to bound.
Every non-branching equation in $E_u\setminus E^\bullet_u$ is 0-1 implied
by $S^\bullet_u$ (but need not be linearly dependent on it over $\F_q$).

\begin{figure}[htbp]
\centering
\begin{tikzpicture}[font=\small, >=Stealth,
  box/.style={draw, rounded corners=3pt, minimum height=9mm,
              align=center, inner sep=5pt}
]

\begin{scope}[xshift=-3.8cm]
  \draw[draw=blue!50, fill=blue!5, rounded corners=5pt]
    (-2.7,-0.6) rectangle (3.0,2.7);
  \node[blue!60!black] at (0,2.1)
    {$E_u$\quad\textit{(all added game equations)}};

  \draw[draw=orange!70, fill=orange!15, rounded corners=3pt]
    (-2.5,-0.3) rectangle (0.4,1.5);
  \node[align=center, orange!70!black] at (-1.05, 0.6)
    {$E^\bullet_u$\\[2pt]
     {\scriptsize branching eqs only}\\[1pt]
     {\scriptsize $\dim E^\bullet_u = \#\text{br.\ points}$}};

  \draw[draw=green!60!black, fill=green!8, rounded corners=3pt]
    (0.5,-0.3) rectangle (2.8,1.5);
  \node[align=center, green!50!black] at (1.6,0.6)
    {non-branching\\[2pt]
     {\scriptsize $h = a$}\\[1pt]
     {\scriptsize $S^\bullet_u\!\models_{0,1}\!h{=}a$}};
\end{scope}

\draw[->, thick] (-0.05,0.8) -- (-0.05,-0.5)
  node[midway, right, font=\scriptsize, align=left]
    {weight~$\leq\tau$\\[-1pt]closure \\ with Ax=b};

\begin{scope}[xshift=3.6cm]
  \draw[draw=blue!50, fill=blue!8, dashed, rounded corners=4pt]
    (-1.5,-1.5) rectangle (3.8,0.2);
  \node[blue!60!black, font=\scriptsize] at (1.2,-1.25)
    {$S_u$\quad (all short equations, weight~$\leq\tau$)};

  \draw[draw=orange!70, fill=orange!20, rounded corners=3pt]
    (-1.0,-1.0) rectangle (1.9,-0.15);
  \node[orange!70!black, align=center] at (0,-0.58)
    {$S^\bullet_u$};


\end{scope}

\draw[->, dashed, gray]
  (-0.7,-0.55) to[out=-30, in=150] (1.9,-1.0);

\node[draw, rounded corners=3pt, fill=gray!8, align=center,
      minimum width=10cm, font=\scriptsize]
  at (-0.05,-2.3)
  {\textbf{Need to show (while $\dim E^\bullet_u < s$, $\tau'=\Omega(\tau)$):}\\
   (1)~$\tau'$-weight equations in $\mathrm{span}(A\cdot x{=}b,E_u)$ are 0-1 satisfiable
    \\ (2)~$S^\bullet_u$ stays 0-1 satisfiable after each branch};
\end{tikzpicture}
\caption{Structure of the modified strategy.  \emph{Left:} $E_u$ splits into
branching equations $E^\bullet_u$ (orange; $\dim E^\bullet_u =$ branching count,
directly controllable) and non-branching equations (green; each is 0-1 implied
by $S^\bullet_u$).
\emph{Right:} $S^\bullet_u \subseteq S_u$ after taking short-weight closures.
The strategy tracks only $S^\bullet_u$ and must establish
properties~\ref{prop1}--\ref{prop2} to guarantee the Delayer is never
prematurely stuck.}
\label{fig:modified-strategy}
\end{figure}

It is easy to see that $S^\bullet_u\models_{0,1} h=a$ for every non-branching
equation $h=a$ in $E_u$.
We need to show that while $\dim(E^\bullet_u)<s:=\Theta\bigl(((q{+}1)\ln q)^{-1/3}d_A^{0.2}\bigr)$:
\begin{enumerate}
\item\label{prop1} For some $\tau' = \Omega(\tau)$, $\tau'$-weight equations in
 $\mathrm{span}(S_u)$ are 0-1 satisfiable. \textit{(Prevents premature endgame; the
  modified strategy only directly guarantees $S^\bullet_u$ is 0-1 satisfiable.)}
\item\label{prop2} After adding a branching equation $f=\tilde{c}$, the updated
  $S^\bullet_{u'}$ is 0-1 satisfiable.
\end{enumerate}

Property~\ref{prop1} ensures Delayer is never forced into a premature endgame
position, and property~\ref{prop2} ensures each branching step preserves
0-1 satisfiability of the tracked system. See Figure~\ref{fig:modified-strategy}.

\medskip
\noindent\textbf{Proof of property~\ref{prop2}.}
Apply Claim~\ref{clm:narrowSat} and the conclusion just after it with $S^\bullet_u$
in place of $S_u$.

\medskip
\noindent\textbf{Proof of property~\ref{prop1}.}
The equations in $E_u\setminus E^\bullet_u$ are 0-1 entailed by $S^\bullet_u$
but may not be $\F_q$-entailed (i.e.\ not in the linear span of $S^\bullet_u$).

\emph{Step 1: find a short partial assignment.}
By the same argument as in Claim~\ref{clm:narrowSat}, there exists $\rho$
satisfying all weight-$<\tau_0$ equations in $\mathrm{span}(S^\bullet_u)$ (and we can show that in fact $|\mathrm{supp}(\rho)|\leq 0.5\tau$).

\emph{Step 2: 0-1 entailment becomes linear dependence.}
After applying $\rho$, the restricted system $S^\bullet_u\rst_\rho$ has large
minimal distance. Theorem~\ref{implClaim} then implies that 0-1 entailment and
linear dependence coincide for $S^\bullet_u\rst_\rho$:

\begin{theorem}\label{implClaim}
For every system $A\cdot x=b$ with $k$ equations and
$d_A\geq 6((q{+}1)\ln q)k^3$, and every equation $f=a$,
if $A\cdot x=b\bmodels f=a$ then $f=a$ lies in the linear span of $A\cdot x=b$.
\end{theorem}

\emph{Step 3: conclude 0-1 satisfiability.}
Each non-branching equation $h=a$ in $E_u$ satisfies
$S^\bullet_u\models_{0,1}h=a$, so $(h=a)\rst_\rho$ lies in
$\mathrm{span}(S^\bullet_u\rst_\rho)$ by Theorem~\ref{implClaim}.
It is not hard to see that this implies that $\leq 0.5\tau$-weight equations in 
$\mathrm{span}(S_u\rst_\rho)$ lie in $\mathrm{span}(S^\bullet_u\rst_\rho)$ (detailed
argument is in the next section).
Since $\mathrm{span}(S^\bullet_u\rst_\rho)$ is
0-1 satisfiable this $\leq 0.5\tau$-weight fragment of $\mathrm{span}(S_u)$ is also
0-1 satisfiable.

This finishes the argument showing $\Omega(((q{+}1)\ln q)^{-1/3}d_A^{1/5})$ lower bound for the number of
branching points.

\subsubsection{Lower bound}\label{sec:TLproofs}

In section we give a detailed proof of the lower bound. Since our arguments are rather
technical, we switch to the notation outlined in Section~\ref{sec:notation} for convenience.

We first derive Theorem~\ref{implClaim} from 
Theorem~\ref{thm:isemptyECC}.(\ref{itm:empECC}).

\begin{proof}(of Theorem~\ref{implClaim})
Let $P$ be a vector space of linear equations and let $k := dim(P)$. Fix some $h=a$. Denote $\tau := 6((q+1)\ln q) k^3$ and assume $\omega(P)\geq \tau$ and 
$P\bmodels h = a$. Recall from Section~\ref{sec:notation} that $red_P(h = a)$ is
some equation $h' = a'$ of minimal weight in 
$\{\alpha h + f\, |\, \alpha\neq 0, f \in \spanv{P}\}$. Consider two cases:
\begin{itemize}
\item $\omega(red_P(h = a)) \geq 0.5\cdot \tau$

By Theorem~\ref{thm:isemptyECC}.(\ref{itm:empECC}) it follows that $P+\langle h = b \rangle$ is 0-1 satisfiable for all $b$, therefore $P\not\bmodels h = a$ which
is a contradiction.

\item $\omega(red_P(h = a)) < 0.5\cdot \tau$

If $red_P(h = a) \equiv (h' = a')\not\equiv (0 = 0)$, we can choose a 0-1 assignment $\rho$ for $<0.5\cdot\tau$ variables such that $h'\rst_{\rho}\equiv b$
for some $b\neq a'$. By Theorem~\ref{thm:isemptyECC}.(\ref{itm:empECC}) it follows that
$(P+\spanv{ h = a - a' + b})\rst_{\rho}$ is 0-1 satisfiable and therefore $P\not\bmodels h = a$ which is a 
contradiction. 

It thus must be that $red_P(h = a)\equiv (0 = 0)$.
\end{itemize}
\end{proof}

We now prove the main theorem for tree-like lower bounds.

\begin{theorem}\label{secondMain}
For all $(A, b) \in \ECC^{n, k, d}_{\F_q}$ there exists Delayer's strategy on $A\cdot x = b$ that guarantees $\Omega({((q+1)\ln q)^{-1/3}}d^{1/5})$ branching points.
\end{theorem}
\begin{proof}

Denote $A\cdot x = b$ equations as $\mathcal{F}=\{f_1 = a_1, \ldots, f_m = a_m\}$.

Let $\tau := d^{4/5}$. Recall from Section~\ref{sec:notation} that $red_P(h = a)$ is
some equation $h' = a'$ of minimal weight in 
$\{\alpha h + f\, |\, \alpha\neq 0, f \in \spanv{P}\}$. We now restate Delayer's strategy
from the beginning of the section more formally.

\bigskip

\noindent\textbf{Delayer's strategy:}

\medskip

Assume the equations added by the current game node are $\mathcal{H}=\{h_1=c_1,\ldots,h_t=c_t\}$
and the linear form chosen by Prover is $l$. Assume moreover that $\mathcal{G}=\{g_1=b_1, \ldots, g_s=b_s\}\subseteq \mathcal{H}$ correspond precisely to 
branching equations, that is to equations added at branching points. Note that 
$\mc H$ can be significantly bigger than $\mc G$ and initially $\mc H$ can be empty.

Denote $(l' = e) := red_{\mathcal{F} + \mathcal{G}}(l = 0)$ and assume $\alpha l' \equiv f + l$, where $f = \alpha e \in \mathcal{F} + \mathcal{G}$, $\alpha\neq 0, e \in \F_q$. 
The strategy is as follows:

\begin{enumerate}
\item\label{nonbranch}\textbf{Non-branching case.} If $[\mathcal{F}+\mathcal{G}]_{\omega \leq \tau} \bmodels l' = c$ for some $c\in \F_q$, then just proceed along the edge 
$l = \alpha c - \alpha e$. The set $\mc H$ is extended with the equation 
$l = \alpha c - \alpha e$.
\item\label{branch}\textbf{Branching case.} Otherwise choose $c_1 \neq c_2\in \F_q$ such that 
$[\mathcal{F}+\mathcal{G}]_{\omega \leq \tau}\cup \{l' = c_i\}$ is 0-1 satisfiable for
$i=1,2$, mark the current node as branching node and proceed in both directions along the edges $l = \alpha c_i - \alpha e$. In branch $i=1,2$ the set $\mc G$ is extended with 
the branching equation $l = \alpha c_i - \alpha e$.
\end{enumerate}

\bigskip

Recall that $[\mathcal{F}+\mathcal{G}]_{\omega \leq \tau}$ is the subspace of 
$\mc F + \mc G$ generated by vectors of weight at most $\tau$. This definition of the strategy ensures that certain invariants hold which in turn 
implies the lower bound.

\begin{claim}\textbf{(Strategy invariants)}\label{clm:invariants} \\ The following invariants hold if 
$s := |\mc G| < 0.5\cdot {(6(q+1)\ln q)^{-1/3}}\cdot d^{1/5}$

\begin{itemize}
\item 
$[\mathcal{F}+\mathcal{G}]_{\omega \leq \tau}\models_{0,1} \mathcal{H}\setminus \mathcal{G}$ where semantic implication is over 0-1 assignments.
\item The system equations $[\mathcal{F}+\mathcal{G}]_{\omega \leq \tau}$ is 
0-1 satisfiable.
\end{itemize}
\end{claim}

\textbf{Invariants imply the lower bound.} Let us first show that if these invariants are
preserved, then $[\mathcal{F}+\mathcal{H}]_{\omega \leq \tau}$ remains 0-1 satisfiable, ensuring that the current position is not an endgame position. The position is not
an endgame position in this case $0 = 1$ is not in $\mathcal{F}+\mathcal{H}$. This, in turn, guarantees that the strategy yields at least $0.5\cdot {(6(q+1)\ln q)^{-1/3}}\cdot d^{0.2}$ branching points. We show this in two steps. 

In the first step, using a satisfying 0-1 assignment $\rho_0$ for $[\mathcal{F}+\mathcal{G}]_{\omega \leq \tau}$, we define a partial assignment $\rho\subset\rho_0$ with 
$|supp(\rho)| \leq 0.5\cdot \tau$
that satisfies all narrow equations in $\mathcal{F}+\mathcal{G}$ in the sense that $\omega((\mathcal{F}+\mathcal{G})\rst_{\rho}) \geq \tau_0$,\footnote{All short equations
are turned by $\rho$ into $0 = 0$. Remaining equations have weight at least $\tau_0$.} where we set $\tau_0 := d^{3/5}$. Note
that $([\mathcal{F}+\mathcal{G}]_{\omega \leq \tau})\rst_{\rho}$ is still 0-1 satisfiable
since $\rho\subset\rho_0$ and $\rho_0$ is a satisfying 0-1 assignment.

In the second step we argue that for such an assignment $\rho$ we have $(\mathcal{H}\setminus \mathcal{G})\rst_{\rho} \subset ([\mathcal{F}+\mathcal{G}]_{\omega \leq \tau})\rst_{\rho}$
meaning that the application of $\rho$ essentially eliminates $\mathcal{H}\setminus \mathcal{G}$ modulo $[\mathcal{F}+\mathcal{G}]_{\omega \leq \tau}$. This follows from
that fact that 
$([\mathcal{F}+\mathcal{G}]_{\omega \leq \tau})\rst_{\rho}$ does not contain narrow equations. Therefore, by Corollary~\ref{implClaim}, whenever 
$([\mathcal{F}+\mathcal{G}]_{\omega \leq \tau})\rst_{\rho}\bmodels h = a$ it must be that $h = a \in ([\mathcal{F}+\mathcal{G}]_{\omega \leq \tau})\rst_{\rho}$.

As we will show, by the second step we have 
$([\mathcal{F}+\mathcal{H}]_{\omega\leq 0.5\cdot \tau})\rst_{\rho} \subseteq ([\mathcal{F}+\mathcal{G}]_{\omega\leq \tau})\rst_{\rho}$,
by the first step $([\mathcal{F}+\mathcal{G}]_{\omega\leq \tau})\rst_{\rho}$ is 0-1 satisfiable, and therefore $[\mathcal{F}+\mathcal{H}]_{\omega\leq \tau}$ is 0-1 satisfiable. 

For the first step we show the following:

\begin{claim}\label{assignClaim}
Let $\mathcal{P}$ and $\mathcal{R}$ be sets of linear equations, $d_{\mathcal{R}} := dim(\mathcal{R})$, fix some $\tau_0 \in \mathbb{N}$ and let $\rho_0$ be a satisfying assignment (not neccessarily 0-1)
for $[\mathcal{P}+\mathcal{R}]_{\omega \leq 2\cdot d_{\mathcal{R}} \cdot \tau_0}$. If 
$d_{\mathcal{R}} \leq 0.5 \cdot \omega(\mathcal{P}) / \tau_0 - 1$ then there exists a partial assignment $\rho \subset \rho_0$ such that
$|supp(\rho)| \leq d_{\mathcal{R}} \cdot \tau_0$ and $\omega((\mathcal{P}+\mathcal{R})\rst_{\rho}) \geq \tau_0$. 
\end{claim}
\begin{proof}
Define partial assignments $\{\rho^{(i)}\}$ inductively as follows:

\begin{enumerate}
\item $\rho^{(0)}=\emptyset$.
\item Let $\rho^{(i)}$ be the current assignment. If there are no equations in $(\mathcal{P}+\mathcal{R})\rst_{\rho^{(i)}}$ of weight less than $\tau_0$, then the induction
stops and $\rho := \rho^{(i)}$. Otherwise if $\hat{g}_i = \hat{b}_i \in \mathcal{P}+\mathcal{R}$ is such that $\omega(\hat{g}_i\rst_{\rho^{(i)}}) < \tau_0$,
then $\rho^{(i+1)}$ extends $\rho^{(i)}$ by setting variables in $vars(\hat{g}_i\rst_{\rho^{(i)}})$ according to $\rho_0$.
\end{enumerate}	

Let $s$ be the number of steps in the inductive procedure above. We now prove that $s\leq d_{\mathcal{R}}$.

Assume $s > d_{\mathcal{R}}$. It is easy to see that $\{\hat{g}_i\}$ are linearly independent. In particular, at the step $d_{\mathcal{R}} + 1$ there are $d_{\mathcal{R}} + 1$
linearly independent equations $\hat{\mathcal{R}}:=\{\hat{g}_1 = \hat{b}_1, \ldots, \hat{g}_{d_{\mathcal{R}} + 1} = \hat{b}_{d_{\mathcal{R}} + 1}\}$. Since 
$|vars(\hat{g}_1) \cup \dots \cup vars(\hat{g}_{d_{\mathcal{R}} + 1})| < (d_{\mathcal{R}} + 1) \cdot \tau_0 \leq 0.5\cdot \omega(\mathcal{P})$ the intersection 
$\spanv {\mc P}\cap \spanv{\hat{\mc {R}}}$ is zero ($0=0$) and therefore $dim(\mathcal{P}+\hat{\mathcal{R}}) = dim(\mathcal{P}) + dim(\hat{\mathcal{R}}) = dim(\mathcal{P}) + d_{\mathcal{R}} + 1$.
But on the other hand $\mathcal{P}+\hat{\mathcal{R}} = \mathcal{P}+\mathcal{R}$ and $dim(\mathcal{P}+\mathcal{R})\leq dim(\mathcal{P}) + d_{\mathcal{R}}$, which is a contradiction.

As a consequence of $s\leq d_{\mathcal{R}}$ we have $|supp(\rho)| \leq s \cdot \tau_0 \leq d_{\mathcal{R}}\cdot \tau_0$. Also since 
$(d_{\mathcal{R}} + 1)\cdot \tau_0 \leq 0.5\cdot \omega(\mathcal{P})$ the vector space $(\mathcal{P}+\mathcal{R})\rst_{\rho}$ is nonzero and therefore
$\omega((\mathcal{P}+\mathcal{R})\rst_{\rho}) \geq \tau_0$. 
\end{proof}

We use Claim~\ref{assignClaim} with $\mathcal{P} := \mathcal{F}$,  $\mathcal{R} := \mathcal{G}$, $\rho_0$ - a satisfying 0-1 assignment for $[\mathcal{F}+\mathcal{G}]_{\omega \leq \tau}$
and $\tau_0 := d^{3/5}$. By Claim~\ref{assignClaim} there exists a partial assignment $\rho\subset \rho_0$ such that $|supp(\rho)|\leq 0.5 \cdot \tau$ and 
$\omega((\mathcal{F}+\mathcal{G})\rst_{\rho})\geq \tau_0$.

We now turn to the second step and prove that $(\mathcal{H})\rst_{\rho}\subseteq ([\mathcal{F}+\mathcal{G}]_{\omega\leq \tau})\rst_{\rho}$. 
Since $\omega((\mathcal{F}+\mathcal{G})\rst_{\rho})\geq \tau_0$ either $([\mathcal{F}+\mathcal{G}]_{\omega\leq \tau})\rst_{\rho}$ is zero or 
$\omega(([\mathcal{F}+\mathcal{G}]_{\omega\leq \tau})\rst_{\rho})\geq \tau_0$.

If $([\mathcal{F}+\mathcal{G}]_{\omega\leq \tau})\rst_{\rho}$ is zero then $(\mathcal{H})\rst_{\rho}$ is also zero.
Otherwise $\omega(([\mathcal{F}+\mathcal{G}]_{\omega\leq \tau})\rst_{\rho})\geq \tau_0$.
In that case we use Corollary~\ref{implClaim} together with the following technical fact:

\begin{claim}\label{shortDimClaim}
Let $\mathcal{P}$ and $\mathcal{R}$ be sets of linear equations, $d_{\mathcal{R}} := dim(\mathcal{R})$ and fix some $\tau_0 \in \mathbb{N}$. If 
$\omega(\mathcal{P}) > d_{\mathcal{R}} \cdot \tau_0$, then $dim([\mathcal{P}+\mathcal{R}]_{\omega\leq \tau_0})\leq d_{\mathcal{R}}$.
\end{claim}

We use Corollary~\ref{implClaim} with $P := ([\mathcal{F}+\mathcal{G}]_{\omega\leq \tau})\rst_{\rho}$. Note that $dim(P)\leq dim([\mathcal{F}+\mathcal{G}]_{\omega\leq \tau}) \leq dim(\mc G)$
where the second inequality is by Claim~\ref{shortDimClaim} since $\omega(\mathcal{F}) > dim(\mathcal{G}) \cdot \tau$. Also 
$\omega(P) \geq \tau_0 \geq ({\CI})\cdot dim(\mc G)^3 \geq (\CI) \cdot dim(P)^3$.
Therefore by Corollary~\ref{implClaim} we have that
$(\mathcal{H})\rst_{\rho}\subseteq P\subseteq (\mathcal{F}+\mathcal{G})\rst_{\rho}$. 

To see that
$([\mathcal{F}+\mathcal{H}]_{\omega\leq 0.5\cdot \tau})\rst_{\rho} \subseteq ([\mathcal{F}+\mathcal{G}]_{\omega\leq \tau})\rst_{\rho}$ consider 
$f + h \in [\mathcal{F}+\mathcal{H}]_{\omega\leq 0.5\cdot \tau}$, where $f \in \mathcal{F}+\mathcal{G}$ and $h \in \mathcal{H} \setminus \mathcal{G}$. Since
$(\mathcal{H}\setminus \mathcal{G})\rst_{\rho} \subset ([\mathcal{F}+\mathcal{G}]_{\omega \leq \tau})\rst_{\rho}$ we have $(f + h)\rst_{\rho} = (f + f')\rst_{\rho}$ for some
$f' \in \mathcal{F}+\mathcal{G}$. And since $\omega((f + f')\rst_{\rho}) \leq 0.5\cdot \tau$ and $|supp(\rho)| \leq 0.5\cdot \tau$ we have $\omega(f + f') \leq \tau$.


We thus have that  
$([\mathcal{F}+\mathcal{H}]_{\omega\leq 0.5\cdot \tau})\rst_{\rho}=([\mathcal{F}+\mathcal{G}]_{\omega\leq\tau})\rst_{\rho}$ is 0-1 satisfiable and therefore 
$[\mathcal{F}+\mathcal{H}]_{\omega\leq 0.5\cdot \tau}$
is 0-1 satisfiable. The second step is completed.


\bigskip

\textbf{Preservation of invariants.} We now prove that invariants still hold after each step provided $s < 0.5\cdot {(6(q+1)\ln q)^{-1/3}}\cdot d^{1/5}$. In case \ref{nonbranch} $\mathcal{G}$ remains unchanged and the invariants hold for obvious reasons. 

In case \ref{branch} we just need to show that $[\mathcal{F}+\mathcal{G}+\langle l = c_i - \alpha e \rangle]_{\omega \leq \tau}$ is 0-1 satisfiable for $i = 1,2$.
If $\omega(l' = c_i) > \tau$, then $[\mathcal{F}+\mathcal{G}+\langle l = \alpha c_i - \alpha e \rangle]_{\omega \leq \tau} = [\mathcal{F}+\mathcal{G}]_{\omega \leq \tau}$ is 0-1 satisfiable. Otherwise let 
$[\mathcal{F}+\mathcal{G}+\langle l = \alpha c_i - \alpha e \rangle]_{\omega \leq \tau} = \langle [\mathcal{F}+\mathcal{G}]_{\omega \leq \tau}, l' = c_i, h_1 = b_1, \ldots, h_t = b_t \rangle$,
where $t\in \N$ and $\omega(red_{\mathcal{F} + \mathcal{G} + \langle l = \alpha c_i - \alpha e \rangle}(h_j = b_j))\geq 0.5 \cdot \tau$ for $j \in [t]$. 

We use the Claim~\ref{assignClaim} with $\mathcal{P} := \mathcal{F}$,  $\mathcal{R} := \mathcal{G}+\langle l' = c_i \rangle$, $\rho_0$ - a 0-1 satisfying assignment for $[\mathcal{F}+\mathcal{G}]_{\omega \leq \tau} + \langle l' = c_i \rangle$
and $\tau_0 := 0.5 \cdot d^{3/5}$. Note that since $\omega(red_{\mathcal{F} + \mathcal{G} + \langle l = \alpha c_i - \alpha e \rangle}(h_j = b_j))\geq 0.5 \cdot \tau$ we have 
$[\mathcal{F} + \mathcal{G} + \langle l = \alpha c_i - \alpha e\rangle]_{\omega \leq 0.5\cdot \tau - 1} = [[\mathcal{F} + \mathcal{G}]_{\omega \leq \tau} + \langle l' = c_i \rangle]_{\omega\leq 0.5\cdot \tau - 1}$ therefore
$\rho_0$ is a 0-1 satisfying assignment for $[\mathcal{F} + \mathcal{G} + \langle l = c_i - \alpha e \rangle]_{\omega\leq 0.5\cdot \tau - 1}$.
By the Claim~\ref{assignClaim} there exists a partial assignment $\rho\subset \rho_0$ such that 
$|supp(\rho)|\leq dim(\mathcal{R})\cd \tau_0\leq s \cdot \tau_0\leq  0.25 \cdot \tau$ and 
$\omega((\mathcal{F}+\mathcal{G} + \langle l = \alpha c_i - \alpha e \rangle)\rst_{\rho})\geq \tau_0$. Therefore 
\begin{multline*}
\omega(([\mathcal{F}+\mathcal{G} + \langle l = \alpha c_i - \alpha e \rangle]_{\omega \leq \tau})\rst_{\rho})\geq \tau_0 \geq ((q+1)\ln q)\cdot s^3 \geq \\ ((q+1)\ln q)\cdot dim(([\mathcal{F}+\mathcal{G} + \langle l = \alpha c_i - \alpha e \rangle]_{\omega \leq \tau})\rst_{\rho})^3 
\end{multline*}
where the last inequality is by the Claim~\ref{shortDimClaim}. By 
Theorem~\ref{thm:isemptyECC}.(\ref{itm:empECC}) 
$([\mathcal{F}+\mathcal{G} + \langle l = \alpha c_i - \alpha e \rangle]_{\omega \leq \tau})\rst_{\rho}$ is 0-1 satisfiable and therefore $[\mathcal{F}+\mathcal{G} + \langle l = \alpha c_i - \alpha e \rangle]_{\omega \leq \tau}$
is 0-1 satisfiable.
\end{proof}

We obtain the following corollaries using Corollary~\ref{cor:rand_code_dist} and
Proposition~\ref{prop:hermit_code_dist}:

\begin{corollary}\label{linTreesLB}
Let $(A, b)\in\ECC^{n,k,d}_{\F_q}$. Every $\LinTrees_{\F_q}$ refutation of $A\cdot x = b$ is of size $2^{\Omega(((q+1)\ln q)^{-1/3}d^{1/5})}$. In particular, every $\LinTrees_{\F_q}$ refutation
of random instances and Hermitian instances are of size $2^{\Omega(((q+1)\ln q)^{-1/3}n^{1/5})}$.
\end{corollary}

\begin{corollary}\label{tlResLB}
Let $(A, b)\in\ECC^{n,k,d}_{\F_q}$. Every tree-like $\reslin{\F_q}$ refutation of $A\cdot x = b$ is of size $2^{\Omega(((q+1)\ln q)^{-1/3}d^{1/5})}$. In particular, every tree-like $\reslin{\F_q}$ refutation
of random instances and Hermitian instances are of size $2^{\Omega(((q+1)\ln q)^{-1/3}n^{1/5})}$.
\end{corollary}
\begin{proof}
Follows from Theorem~\ref{delayerEquiv} and Theorem~\ref{secondMain}.
\end{proof}

\subsection{Binary regular dag-like $\reslin{\F_q}$ lower bounds}\label{bindagLB}

We sketch the proof informally; full details follow.
Fix a $\BinRegDags_{\F_q}$ refutation $T$ of $(A,b)$.

\medskip\noindent\textbf{Setup.}
At each node $u$ of $T$, write $\mathcal{F}_u$ for the labelling system.
Every 0-1 assignment $\rho$ traces a path $\pi_\rho$ through $T$ from root to a contradiction leaf.
Along $\pi_\rho$, define $\rho_t$ as a minimal-support partial assignment such that
$\mathcal{F}_{v_t} \subseteq \langle\mathcal{F}\rst_{\rho_t}\rangle$,
where $v_t$ is the $t$-th node on $\pi_\rho$.
Let $v_\rho$ be the \emph{checkpoint}: the first node on $\pi_\rho$ at which $|\mathrm{supp}(\rho_t)| = s$.
Write $\hat\rho := \rho_{t_0}$ for the corresponding minimal partial assignment.

At the checkpoint, there exists $(C_\rho\cdot x = d_{\rho}) \subseteq \langle\mathcal{F}\rangle$ (a set of equations
over $A\cdot x = b$'s span) such that $\mathcal{F}_{v_\rho} = (C_\rho\cdot x = d_{\rho})\rst_{\hat\rho}$.
Define the \textbf{label} of $\rho$ as $\ell(\rho) := (C_\rho,\, I_\rho)$ where
$I_\rho := \mathrm{supp}(\hat\rho)$. See Figure~\ref{fig:dag-paths}.

\begin{figure}[htbp]
\centering
\begin{subfigure}[t]{0.47\linewidth}
\centering
\resizebox{\linewidth}{!}{%
\begin{tikzpicture}[
  nd/.style={circle, draw, minimum size=6mm, inner sep=1pt, font=\small},
  chk/.style={rectangle, rounded corners=2pt, draw, thick, fill=yellow!30,
              minimum size=7mm, font=\small},
  lf/.style={draw=none, font=\small},
  p1/.style={blue!70!black, thick},
  p2/.style={red!70!black, thick},
  >=Stealth
]
\node[nd]  (r)   at ( 0,  0  ) {};   
\node[nd]  (a)   at (-2, -1.5) {};
\node[nd]  (c)   at ( 2, -1.5) {};
\node[chk]  (d)   at (-2, -3.0) {$v_{\rho^{(1)}}$};
\node[nd]  (f)   at ( 2, -3.0) {};
\node[chk] (h)   at ( 2, -4.5) {$v_{\rho^{(2)}}$};
\node[lf]  (l1)  at (-2, -6.0) {$\bot$};
\node[lf]  (l3)  at ( 2, -6.0) {$\bot$};
\draw[gray!50] (r)--(c);
\draw[gray!50] (a)--(f);
\draw[gray!50] (c)--(f);
\draw[gray!50] (f)--(h);
\draw[gray!50] (h)--(l3);
\draw[p1] (r) -- (a);
\draw[p1] (a) -- (d);
\draw[p1,dashed] (d) -- (l1) node[midway, right, font=\scriptsize, p1]{after checkpoint};
\draw[p2] (r) -- (c);
\draw[p2] (c) -- (f);
\draw[p2] (f) -- (h);
\draw[p2,dashed] (h) -- (l3);
\node[font=\scriptsize, blue!70!black, right=4pt] at (d.east)
  {$\ell(\rho^{(1)})=(C_1,I_1)$};
\node[font=\scriptsize, blue!70!black, left=4pt] at (d.west)
  {$|\mathrm{supp}(\rho^{(1)}_t)|{=}s$};
\node[font=\scriptsize, blue!70!black, right=4pt] at (a.east)
  {$|\mathrm{supp}(\rho^{(1)}_t)|{<}s$};
\node[font=\scriptsize, red!70!black, right=4pt] at (h.east)
  {$\ell(\rho^{(2)})=(C_2,I_2)$, $|\mathrm{supp}(\rho^{(2)}_t)|{=}s$};
\node[font=\scriptsize, red!70!black, right=4pt] at (f.east)
  {$|\mathrm{supp}(\rho^{(2)}_t)|{<}s$};
\draw[p1] (-6.5,-4.2) -- (-5.5,-4.2) node[right,font=\scriptsize]{path of $\rho^{(1)}$};
\draw[p2] (-6.5,-4.8) -- (-5.5,-4.8) node[right,font=\scriptsize]{path of $\rho^{(2)}$};
\node[chk, font=\scriptsize] at (-6.0,-5.5) {$v$};
\node[font=\scriptsize, right=-1pt] at (-5.5,-5.5) {checkpoint ($|\hat\rho|{=}s$)};
\end{tikzpicture}}
\caption{Each assignment $\rho$ has a checkpoint $v_\rho$ (yellow) where
$|\mathrm{supp}(\hat\rho)|=s$; labelled $\ell(\rho)=(C_\rho,I_\rho)$.}
\label{fig:dag-paths}
\end{subfigure}
\hfill
\begin{subfigure}[t]{0.47\linewidth}
\centering
\resizebox{\linewidth}{!}{%
\begin{tikzpicture}[font=\small, >=Stealth,
  blob/.style={draw, ellipse, minimum width=2.2cm, minimum height=3.5cm,
               align=center, inner sep=4pt},
  lbl/.style={draw, rounded corners=2pt, fill=gray!10, inner sep=3pt,
              minimum width=2cm, align=center, font=\scriptsize}
]
\node[blob, fill=blue!5, label=above:{$X$ (nodes $v_\rho$)}] (X) at (-3.5,0) {};
\foreach \y in {1.0, 0.3, -0.4, -1.1}
  \node[circle, draw, fill=blue!20, minimum size=5mm, inner sep=0] at (-3.5,\y) {};
\node[blob, fill=orange!8, label=above:{$Y$ (labels $\ell(\rho)$)}] (Y) at (3.5,0) {};
\node[lbl] (y1) at (3.5, 0.9) {$(C_1,I_1)$};
\node[lbl] (y2) at (3.5, 0.0) {$(C_2,I_2)$};
\node[lbl] (y3) at (3.5,-0.9) {$\vdots$\;\;$\vdots$};
\draw[->] (-2.8, 0.9) -- (2.5, 0.9);
\draw[->] (-2.8, 0.2) -- (2.5, 0.2);
\draw[->] (-2.8,-0.5) -- (2.5, 0.0);
\draw[->] (-2.8,-1.1) -- (2.5,-0.9);
\node at (0, 1.5) {$F$\quad(surjection)};
\node[draw, fill=green!8, rounded corners=3pt, align=center, font=\scriptsize]
  (c1) at (3.5, -2.5)
  {\textbf{Case 1:} $|Y|\geq 2^{\Omega(r)}$\\ $\Rightarrow |X|\geq|Y|\geq 2^{\Omega(r)}$};
\node[draw, fill=yellow!15, rounded corners=3pt, align=center, font=\scriptsize]
  (c2) at (-3.5,-2.8)
  {\textbf{Case 2:} some $(C,I)$ has\\ large fibre $|F^{-1}(C,I)|$\\
   $(s,r)$-robust: $\mathrm{rank}(C_I)\!\geq\! r$\\
   many distinct $\mathcal{F}_{v_\rho}$ as $\rho$ varies\\
   $\Rightarrow |X|\geq 2^{\Omega(r)}$};
\draw[->, dashed] (c2.north) -- (y2.south west);
\end{tikzpicture}}
\caption{Surjection $F:X\twoheadrightarrow Y$ gives $|X|\geq|Y|$.
Case~1: $|Y|\geq 2^{\Omega(r)}$ directly.
Case~2: large fibre + $(s,r)$-robustness $\Rightarrow$ $2^{\Omega(r)}$ distinct nodes.}
\label{fig:dag-counting}
\end{subfigure}
\caption{Lower bound via a surjection.}
\label{fig:dag-combined}
\end{figure}

\noindent\textbf{Lower bound via a surjection.}
Set $X:=\{v_\rho\}_\rho$ (all checkpoint nodes) and $Y:=\{\ell(\rho)\}_\rho$ (all labels).
The map $F: v_\rho\mapsto\ell(\rho)$ is a well-defined surjection $X\twoheadrightarrow Y$,
so $|X|\geq|Y|$. For some $K = 2^{\Omega(r)}$ we show $|X|\geq K$ by splitting into two cases.

\noindent\textbf{Case~1}: $|Y|\geq K$.
Then $|X|\geq|Y|\geq K$ directly from the surjection.

\noindent\textbf{Case~2}: $|Y| < K$.
By pigeonhole, some label $(C,I)\in Y$ is shared by a large set
$Z:=\{\rho:\ell(\rho)=(C,I)\}$.
All $\rho\in Z$ satisfy $I_\rho=I$, so the partial assignments
$\hat\rho$ vary only on $I$. In particular, $\hat Z:=\{\hat \rho:\rho \in Z\}$ must be large.
Since $(A,b)$ is $(s,r)$-robust, $\mathrm{rank}(C_I)\geq r$, which makes the
map $\hat\rho\mapsto (C\!\cdot\!x)\rst_{\hat\rho}$ injective on $Z$:
distinct $\hat\rho$ produce distinct restricted systems $\mathcal{F}_{v_\rho}$,
hence distinct nodes $v_\rho\in X$.
Hence, the number of distinct $v_\rho\in X$ must be large and in fact one can show that
it is at least $K$.

In both cases $|T|\geq|X|\geq K$, finishing the sketch. See Figure~\ref{fig:dag-counting}.

We now give the full proof.

\begin{theorem}\label{firstMain}
Assume an instance $(A, b) \in \ECC^{n, k, d}_{\F_q}$ is $(s, r)$-robust for some $s < d / 2$.
Recall that $d = d_A := \omega(\mc C_A)$ is the minimal distance of the code
$\mc C_A = \{x\cdot A \mid x\in\F_q^k\}$ (see Section~\ref{sec:notation}).
Then every $\BinRegDags_{\F_q}$ refutation of $(A, b)$ is of size $2^{\Omega(r)}$.
\end{theorem}
\begin{proof}
Let $T$ be a $\BinRegDags_{\F_q}$ refutation of $(A, b)$ instance written as $\mathcal{F}=\{f_1=a_1, \ldots, f_m = a_m\}$.
Every 0-1 assignment $\rho$ defines a path $p_{\rho}=(v_0^{(\rho)}\equiv r, v_1^{(\rho)}, \ldots, v_{l(\rho)}^{(\rho)})$ from the root $r \in T$ to a terminal node: if $v_i^{(\rho)} \in T$ is 
marked with a variable $x_{v_i^{(\rho)}}$ then the edge $(v_i^{(\rho)}, v_{i+1}^{(\rho)})$ is marked with $\rho(x_{v_i^{(\rho)}})$. 

For an initial segment $p_t'=(v_0^{(\rho)}\equiv r, v_1^{(\rho)}, \ldots, v_t^{(\rho)})\subset p_{\rho}$ define a partial assignment $\rho_t\subset \rho$ to be arbitrary partial
assignment of minimal size such that $supp(\rho_t)\subset \{x_{v_0^{(\rho)}}, \ldots, x_{v_t^{(\rho)}}\}$ and 
$\mathcal{F}_{v_t^{(\rho)}} \subseteq \langle \mathcal{F}\rst_{\rho_t}\rangle$. 



Let $s_0$ be the minimal such that 
$|supp(\rho_{s_0})| = s$ and denote $\hat{\rho} := \rho_{s_0}$, $v_{\rho} := v_{s_0}^{(\rho)}$. 
Note that $s_0$ always exists since if $|supp(\rho_{t})| < s$ then 
$\omega(\langle \mathcal{F}_{v_t^{(\rho)}}\rst_{\rho_t}\rangle_{[1]})\geq d - s \geq 1$
(since $\langle \mathcal{F}\rangle_{[1]}$ is ECC) and therefore $v_t^{(\rho)}$ is not a terminal node.

Our goal now is to show that the set $X := \{v_{\rho}\}_{\rho}$ is large.

To an assignment $\rho$ we associate a set of linear equations $\mathcal{F}_{\rho}$ such that $\mathcal{F}_{\rho} \subseteq \langle\mathcal{F}\rangle$ and
$\mathcal{F}_{v_{\rho}} = \mathcal{F}_{\rho}\rst_{\hat{\rho}}$. The existence of such $\mathcal{F}_{\rho}$ follows from definitions of $\hat{\rho}$ and $v_{\rho}$.
Note that the maps $\mathcal{F}_{\rho} \rightarrow \mathcal{F}_{\rho}\rst_{\hat{\rho}}$ and $(\mathcal{F}_{\rho})_{[1]} \rightarrow (\mathcal{F}_{\rho})_{[1]}\rst_{\hat{\rho}}$
given by application of $\hat{\rho}$ are bijective since $|supp(\hat{\rho})| = s < d / 2$.

Consider the set $Y := \{((\mathcal{F}_{\rho})_{[1]}, supp(\hat{\rho}))\}_{\rho}$. 
The set $Y$ is not larger than $X$: inequality $|X|\geq |Y|$ follows from the existence of a surjective function $X \rightarrow Y$.

\begin{claim*} 
The relation $R := \{(v_{\rho}, ((\mathcal{F}_{\rho})_{[1]}, supp(\hat{\rho})))\}_{\rho} \subset X \times Y$ defines a surjective function $F\, :\, X \rightarrow Y$.
\end{claim*}
\begin{proof}
Surjectivity of $F$ is obvious. It remains to see that $R$ defines a function. 

Assume
$((\mathcal{F}_{\rho_1})_{[1]}, supp(\hat{\rho}_1))\neq
 ((\mathcal{F}_{\rho_2})_{[1]}, supp(\hat{\rho}_2))$.
If $supp(\hat{\rho}_1) \neq supp(\hat{\rho}_2)$ then 
because of minimality of $\hat{\rho}_1$ and $\hat{\rho}_2$ it follows that 
$(\mathcal{F}_{\rho_1})_{[1]}\rst_{\hat{\rho}_1} \neq (\mathcal{F}_{\rho_2})_{[1]}\rst_{\hat{\rho}_2}$ and therefore $v_{\rho_1}\neq v_{\rho_2}$. If $supp(\hat{\rho}_1) = supp(\hat{\rho}_2)$ then 
$(\mathcal{F}_{\rho_1})_{[1]} \neq (\mathcal{F}_{\rho_2})_{[1]}$ and therefore 
$(\mathcal{F}_{\rho_1})_{[1]}\rst_{\hat{\rho}_1} \neq (\mathcal{F}_{\rho_2})_{[1]}\rst_{\hat{\rho}_2}$ 
by injectivity of the map $(\mathcal{F}_{\rho})_{[1]} \rightarrow (\mathcal{F}_{\rho})_{[1]}\rst_{\hat{\rho}}$ and thus $v_{\rho_1}\neq v_{\rho_2}$.

\end{proof}

For $\epsilon \in (0, 1)$ consider two cases: $|Y| \geq 2^{\epsilon\cdot s}$ and $|Y| < 2^{\epsilon\cdot s}$.

\begin{itemize}
\item \textbf{Case $|Y| \geq 2^{\epsilon\cdot s}$.} Since $|X| \geq |Y|$ in this case we have immediately a lower bound $|X| \geq 2^{\epsilon\cdot s}$.

\medskip
 \item \textbf{Case $|Y| < 2^{\epsilon\cdot s}$.} In this case there exists $(F_0, I_0)\in Y$ such that the set of assignments
$A_0 := \{ \rho\, | \, F_{\rho} = F_0, supp(\hat{\rho}) = I_0\}$ has cardinality at least $2^{n - \epsilon\cdot s}$. Note that since linear forms in $\mathcal{F}$ are linearly independent
there exists unique $\mathcal{F}_0$ such that $\mathcal{F}_0 \subseteq \langle \mathcal{F} \rangle$ and $\langle \mathcal{F}_0 \rangle_{[1]} = \langle F_0 \rangle$.

For a partial assignment $\rho_0$ with $supp(\rho_0) = I_0$ there exist at most $2^{n - s}$ extensions to a full assignment. Therefore the set of partial assignments 
$B_0:= \{ \hat{\rho}\, | \, \rho \in A_0\}$ has cardinality at least $|A_0| / 2^{n - s} \geq 2^{(1 - \epsilon)\cdot s}$. Let $J := vars(F_0)$, $l := |F_0|$ and let $M$ be a $l \times |J|$ matrix and 
$b\in \mathbb{F}_q^l$ be a vector such that $M\cdot x = b$ is the system $\mathcal{F}_0$. For an assignment $\rho\in A_0$ the system 
$\mathcal{F}_{v_{\rho}}$, which coincides with $\mathcal{F}_0\rst_{\hat{\rho}}$, can be written as 
$M_{[J\setminus I_0]}\cdot x_{[J\setminus I_0]} = b - M_{[I_0]}\cdot a_{\hat{\rho}}$ where $a_{\hat{\rho}}\in \mathbb{F}_q^s$ is the vector of values assigned by $\hat{\rho}$. 
Therefore for two assignments $\rho_1, \rho_2\in A_0$ the systems $\mathcal{F}_{v_{\rho_1}}$ and $\mathcal{F}_{v_{\rho_2}}$ are different and thus $v_{\rho_1} \neq v_{\rho_2}$
iff $M_{[I_0]}\cdot a_{\hat{\rho}_1} \neq M_{[I_0]}\cdot a_{\hat{\rho}_2}$.

Consider the set $C_0 := \{M_{[I_0]}\cdot a_{\hat{\rho}}\, |\, \rho \in A_0\}$. From the argument above it follows that $|X| \geq |C_0|$. We thus need to lower bound the image
of the set $B_0 = \{a_{\hat{\rho}}\, | \, \rho \in A_0\}$ under the operator given by $l\times s$ matrix $M_{[I_0]}$. We will use two properties of $B_0$ and $M_{[I_0]}$:
$|B_0| \geq 2^{(1-\epsilon)\cdot s}$ and $rank(M_{[I_0]})\geq r$ 
(by $(s,r)$-robustness).

\begin{claim*}\label{imgLB}
Let $M$ be $k\times n$ matrix over $\mathbb{F}_q$, $r := rank(M)$ and $X \subseteq \mathbb{F}_q^n$ be such that $|X| \geq 2^{(1 - \epsilon)\cdot n}$ for some
$\epsilon \in (0, 1)$. Then $|M(X)| \geq 2^{r - \epsilon\cdot n}$.
\end{claim*}
\begin{proof}
Note that performing Gaussian elimination on $M$ does not change $|M(X)|$. Arrange matrix $M$ in block form as follows (possibly applying Gaussian elimination):
\[
  M = \kbordermatrix{
    & r &  & n-r  \\
    r & I_r & \vrule & *  \\\cline{2-4}
    k-r & 0 & \vrule & * 
  }
\]

Every vector $v\in \mathbb{F}_q^n$ we split accordingly $v = a_v b_v$, where $a_v$ contains first $r$ coordinates and $b_v$ last $n - r$ coordinates. Since there are 
$2^{n - r}$ suffixes $b_v$ it follows that there exists $b_0$ such that the set $X_0 := \{v \, | \, v\in X, b_v = b_0\}$ is of size $|X_0|\geq |X| / 2^{n - r} \geq 2^{r - \epsilon\cdot n}$.
The injectivity of $M$ on $X_0$ implies the bound: $|M(X)| \geq |M(X_0)| = |X_0|\geq 2^{r - \epsilon\cdot n}$.
\end{proof}
By the Claim~\ref{imgLB} we have $|X| \geq 2^{r - \epsilon\cdot s}$.
\end{itemize}
The case analysis above shows $|X| \geq \min(2^{\epsilon\cdot s}, 2^{r - \epsilon\cdot s})$. Choosing $\epsilon := r / 2s$ we obtain $|X| \geq 2^{r / 2}$.
\end{proof}

\begin{corollary}\label{cor:dag-LB-rand}
If $A$ is uniformly random $k \times n$ matrix over $\F_q$ where 
$n / \log q < k < n / 9$ 
then there exists $b\in \F_q^k$ such that $A\cdot x = b$ is 0-1 unsatisfiable and
with high probability $\BinRegDags_{\F_q}$ all refutatioins of $A\cdot x = b$ are
of size $2^{\Omega\left((n/(q + 1)\ln q)^{1/3}\right)}$.
\end{corollary}

\begin{corollary}\label{cor:dag-LB-herm}
Let $A$ be a generator matrix of the Hermitian code over $\F_q$
for $q = p^2$ for a prime $p$, with parameters $n = p^3$, $k = \lfloor p^3/\log_2 p - p^2 / 2 \rfloor$,
$d \geq (1 - 1/\log p) n$. Then there exists $b\in \F_q^k$ such that $A\cdot x = b$
is 0-1 unsatisfiable and such that all $\BinRegDags_{\F_q}$ refutations
are of size $2^{\Omega\left((n/(q + 1)\ln q)^{1/3}\right)}$.
\end{corollary}

\section{Conclusion}\label{sec:conclusion}

The results in this paper open several directions for future research.

\begin{itemize}
\item Extend Theorem~\ref{firstMain} to work with arbitrary linear forms, i.e.\ to
  general dag-like $\reslin{\F_q}$ refutations.
\item Extend the lower bound of Theorem~\ref{firstMain} to a strengthening of
  $\BinRegDags_{\F_q}$ obtained by relaxing the regularity condition.
\item Determine whether ECC distance or $(s,r)$-robustness provide a measure of
  hardness for Nullstellensatz or Polynomial Calculus, or identify another natural
  parameter that does.
\item Narrow the gap $[\Omega(k\log q),\;((q+1)\ln q)k^3]$ for the threshold
  $\Delta(k,q)$ above which every system $A\cdot x=b$ with $d_A\geq\Delta(k,q)$
  is guaranteed to be 0-1 satisfiable.
\item Determine whether $(s,\omega(\log n))$-robust instances exist in
  $\ECC^{n,k,d}_{\F_q}$ for $s\leq k/2$.
\item Find a natural 0-1 unsatisfiable system $A\cdot x=b$ encoding a combinatorial
  principle (e.g.\ the Pigeonhole Principle) where $A$ generates a good ECC.
\item Find an explicit $b\in\F_q^k$ with $b\notin A_{RS}(\{0,1\}^{q-1})$ for the
  Reed-Solomon code, or prove no such $b$ exists.
\end{itemize}

\bibliographystyle{plain}
\bibliography{PrfCmplx-Bakoma}

\end{document}